%% file: multicast-unicast.tex
\newif\ifreport\reporttrue

\ifreport
\documentclass[conference]{IEEEtran}
\else
\documentclass[conference]{IEEEtran}
\fi
\IEEEoverridecommandlockouts

\makeatletter
\def\ps@headings{%
\def\@oddhead{\mbox{}\scriptsize\rightmark \hfil \thepage}%
\def\@evenhead{\scriptsize\thepage \hfil \leftmark\mbox{}}%
\def\@oddfoot{}%
\def\@evenfoot{}}
\makeatother
\usepackage{amsfonts}
\usepackage{amsthm}
\usepackage{amsmath}
\usepackage{amssymb}
\usepackage{bm}
\usepackage{cite}
\usepackage{graphicx}
\usepackage[tight,footnotesize]{subfigure}
\newtheorem{lemma}{Lemma}

\newtheorem{Theorem}{Theorem}

\makeatletter
\renewcommand{\maketag@@@}[1]{\hbox{\m@th\normalsize\normalfont#1}}%
\makeatother

\DeclareMathOperator*{\argmax}{argmax}

\usepackage{graphics,booktabs,color,epsfig,subfigure}

\begin{document}
\IEEEoverridecommandlockouts

\title{Scheduling of Multicast and Unicast Services under Limited Feedback by using Rateless Codes}

\author{Yin Sun$^\dag$, C. Emre Koksal$^\dag$, Kyu-Han Kim$^\ddag$, and Ness B. Shroff$^\dag{}^\S$\\
$^\dag$Dept. of ECE, $^\S$Dept. of CSE, The Ohio State University, Columbus, OH\\
$^\ddag$Hewlett Packard Laboratories, Palo Alto, CA
\thanks{This work has been supported in part by an IRP grant from HP.}
}


\maketitle

\input{sec_abstract}

%
%
%
%
%

\input{sec1}

\input{sec2}

\input{sec3}

\input{sec4}
\input{sec5}

\input{sec7}
\input{sec8}

\input{sec9_conclusion}


\bibliographystyle{IEEEtran}
\bibliography{multicast}
\appendices
\input{appendex0}
\ifreport
\input{appendex}

\input{appendex1}

\fi

\end{document}

%% file: sec_abstract.tex
\begin{abstract}


Many opportunistic scheduling techniques are impractical because they require accurate channel state information (CSI) at the transmitter. In this paper, we investigate the scheduling of unicast and multicast services in a downlink network with a very limited amount of feedback information. Specifically, unicast users send imperfect (or no) CSI and infrequent acknowledgements (ACKs) to a base station, and multicast users only report infrequent ACKs to avoid feedback implosion. We consider the use of physical-layer rateless codes, which not only combats channel uncertainty, but also reduces the overhead of ACK feedback. A joint scheduling and power allocation scheme is developed to realize multiuser diversity gain for unicast service and multicast gain for multicast service. We prove that our scheme achieves a near-optimal throughput region. Our simulation results show that our scheme significantly improves the network throughput over schemes employing fixed-rate codes or using only unicast communications.
\vspace{-0.3cm}
\end{abstract}

%% file: sec1.tex
\section{Introduction}\label{sec:intro}
Over the past decade, opportunistic scheduling techniques have been developed to improve the throughput of wireless networks. Many of these techniques require accurate channel state information (CSI) at the transmitter. However, obtaining this information is costly and could incur significant feedback overhead in practical wireless networks \cite{Love2008}. This issue is particularly serious for multicast services, where the number of users could be large and sending back each user's channel state may result in feedback implosion \cite{Anastasopoulos2012}.
Therefore, a key question is ``how to optimally manage the network resources and exploit the most from a limited amount of feedback information?''

In this paper, we aim to answer this question by jointly considering channel coding, scheduling, and power allocation for a downlink network with both unicast and multicast services, which is illustrated in Fig. \ref{fig11}. {In this network, the power and channel resources are shared among the unicast and multicast service flows. In addition, the throughput of multicast service can be enhanced by virtue of unicast retransmission after the multicast session \cite{MBMS_LTEA2012}.} The network adopts a practical limited feedback mechanism that is suggested by the LTE-Advanced MBMS standards \cite{Oyman 2010, Gruber2011,MBMS_LTEA2012}: The unicast users send imperfect CSI and report one-bit acknowledgements (ACKs) to the base station \cite{Oyman 2010}. However, the multicast users are only allowed to send ACKs to the base station\footnote{The ACK signaling procedure for multicast service is realized by a reception reporting mechanism defined in the LTE-Advanced MBMS standards \cite{MBMS_LTEA2012}.} without feeding back any CSI\footnote{Reporting the CSI of each multicast user is usually inefficient, because the feedback overhead could be enormous and the throughput of multicast service is determined by only the channel condition of the worst-case user \cite{Oyman 2010,Gruber2011}.}.

\begin{figure}
\centering
\includegraphics[width=2.2in]{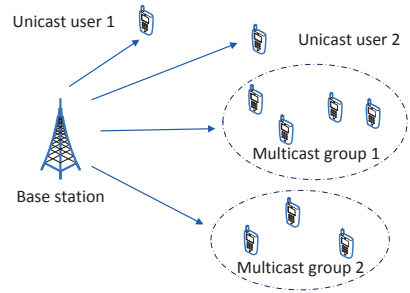}
\vspace{-0cm}
\caption{System model of the downlink network.} \label{fig11}
\vspace{-0.5cm}
\end{figure}

An important part of our answer lies in the adoption of \emph{rateless codes} \cite{Draper09}. Traditional fixed-rate codes, such as Turbo codes and LDPC codes, suffer from rate loss if there is channel uncertainty. Specifically, if the code rate is higher than the capacity, a decoding error occurs; if the code rate is lower than the capacity, the available capacity is not fully utilized. In contrast, rateless codes leverage automatic rate adaptation to combat channel uncertainty. As the transmitter progressively sends the coded packets of a rateless code to the receiver, the code rate decreases over time. Once the code rate drops below the empirical channel capacity, the receiver can decode the message and send an ACK to the transmitter. Therefore, the throughput of rateless codes is close to the empirical channel capacity \cite{Draper09}. Furthermore, rateless codes can significantly reduce the amount of ACK feedback \cite{Eswaran10}, because, unlike fixed-rate codes, they do not require an ACK/NACK for each packet.
\IEEEpubidadjcol

We develop a scheduling and power allocation scheme for unicast and multicast services, where rateless codes (e.g., Raptor codes \cite{Etesami06}, Strider \cite{Gudipati2011}, and Spinal codes \cite{Balakrishnan2012}) are employed in the physical layer to replace fixed-rate codes. While rateless codes can improve the network throughput and reduce ACK feedback, \emph{they significantly complicate the network control problem}. In rateless codes, since the transmission of one message could span over a random number of slots, message decoding and queue updates occur intermittently at irregular time intervals. However, the network control action needs to vary in each slot to achieve multiuser diversity gain for unicast users. Therefore, the scheduler may switch to another user with better channel quality before the previous user has decoded its message and transmitted back an ACK to update the queue state. \emph{This is quite different from traditional queueing systems for fixed-rate codes, e.g., \cite{NetworkResourceAllocationBook06,Lin06,Zhao_Jsac07,Murugesan12,Wenzhuo12}, where the control action and queue lengths are both updated periodically in each slot.}

This network control problem is challenging for two reasons: 1) the transmission procedure of rateless codes is correlated over time; 2) the network controller has very limited channel knowledge to decide the control action. In particular, traditional \emph{slot-level} Lyapunov drift techniques, e.g., \cite{NetworkResourceAllocationBook06,Neelybook10,Sun2013}, are not sufficient to establish the stability of our network control scheme based on rateless codes. To that end, the following are the contributions of our paper:
\begin{itemize}
\item We develop a low-complexity scheduling and power allocation scheme for rateless codes. Our scheme can fully exploit the imperfect CSI and infrequent ACK feedback to simultaneously realize multiuser diversity gain and multicast gain. Moreover, our scheme is quite robust against channel uncertainty. To the best of our knowledge, \emph{our network control scheme with rateless codes is the first that can simultaneously realize multiuser diversity gain, multicast gain, and robustness against channel uncertainty under limited feedback information.}

\item We show that our scheme can achieve a near-optimal stability region. In doing so, we utilize fluid limit techniques to resolve the temporal correlation in the transmission procedure of rateless codes.

\item To prevent one or several multicast users with poor channel conditions to bring down the throughput of the multicast group, we divide the multicast users according to their channel qualities and set up separate unicast sessions sending additional data to the users with bad channels. By this, the throughput of multicast communications is enhanced. The stability region of this combined delivery scheme is attained.

\end{itemize}

Before we proceed, we would like to emphasize that we are using rateless codes at the physical layer, as opposed to those that address the packet erasures (i.e., packet drops/errors) at the application layer, with fixed-rate codes being used in the physical layer, e.g, \cite{Calabuig2013}. In this setting, if one packet is not successfully decoded in the physical layer, it is discarded and does not contribute toward decoding of the application layer message. On the other hand, with physical layer rateless codes, each packet contributes toward information accumulation, even if the packet is not decodable by itself. The throughput performance of these two approaches is compared in Section \ref{sec:simulation}. 

%% file: sec2.tex
\section{Related Work}
Opportunistic scheduling of unicast services based on rateless codes
have been extensively studied based on the assumption of accurate CSI (e.g., see the survey paper \cite{Lin06} and the references therein). Under channel uncertainty, ARQ-based feedback mechanisms for retrieving partial CSI have been developed for ON/OFF Markov channels, e.g., \ifreport\cite{Zhao_Jsac07,Murugesan12,Wenzhuo12,Li2013528}. 
\fi
These mechanisms rely on the time-correlation in the channels to retrieve knowledge of the channel states, and thus do not work well when the temporal correlation is low.

For multicast services, scheduling based on fixed-rate codes has been investigated in, e.g.,
\ifreport
\cite{Won09,Pantelidou10,Pantelidou11,Tsai11}.
\fi
These schemes require the CSI of all the multicast users, which may result in a huge amount of feedback overhead. Without accurate CSI, such schemes may suffer from a significant throughput loss.

Scheduling of multicast transmissions using rateless codes was studied in \cite{Low10}, where a subset of multicast users were scheduled in each slot and the average system delay was minimized by resorting to extreme value theory. A joint scheduling and code length adaptation approach was developed for multicast services with hard deadlines by using dynamic programming \cite{Yang:2012}. Practical MAC protocols for WiFi networks based on rateless codes were proposed by virtue of dynamic programming \cite{Iannucci:2012} and channel prediction \cite{Gudipati:2012}. However, power allocation was not incorporated in these studies due to the inherent computational complexity. The throughput performance of multicast transmissions using rateless codes was studied in, e.g.,
\ifreport
\cite{Cogill2011,Yang2012,Swapna2013}.
\fi
However, these results only apply to erasure channels and are difficult to be extended to general fading channels.
A joint power allocation, scheduling, and message size adaptation approach was proposed for unicast service with rateless codes in \cite{Sun2013}. The analysis there cannot be extended to scenarios with multicast service; see Section \ref{sec:throughput1} for more details.

%% file: sec3.tex
\section{System Model}\label{sec2}
We consider a time-slotted downlink cellular network with one base station, $U$ unicast users, and $G$ multicast user groups, as illustrated in Fig. \ref{fig11}. In each time slot, the base station can schedule only one unicast or multicast data flow. Let $u\in\{1,2,\cdots,U\}$ be the flow index of the unicast users and $g\in\{U+1,2,\cdots,U+G\}$ be the flow index of the multicast user groups. The scheduled flow at slot $t$ is denoted by $s(t)\in\mathcal {S}\triangleq\{1,2,\cdots,G+U\}$.

The channels are assumed to be block fading with a constant channel state within each slot, and vary from one slot to another.~Let $h_u(t)$ represent the channel state of unicast user $u$ at slot $t$ and $h_{gj}(t)$ represent the channel state of multicast user $j$ in the $g$th group with $j\in\{1,\cdots, J(g)\}$, where $J(g)$ is the number of  multicast users in the $g$th group. Each user has perfect knowledge of its downlink CSI via channel estimation. However, the base station only has access to an imperfect CSI
$\hat{\textbf{h}}(t)\triangleq\{\hat{h}_1(t),\cdots,\hat{h}_U(t)\}$ for the unicast users and the conditional probability distribution $f({h}_u|\hat{{h}}_u)$ of ${h}_u(t)$ for given value of $\hat{{h}}_u(t)$ \cite{Love2008,Goldsmith2003}, with no channel knowledge for the multicast users \cite{Oyman 2010}. This model has covered the special cases of no CSI feedback, i.e., $\hat{{h}}_u(t)$ is independent of ${h}_u(t)$, and perfect CSI feedback, i.e., $\hat{{h}}_u(t)={h}_u(t)$. We assume that $\{{h}_u(t), \hat{h}_u(t)\}$ and $h_{gj}(t)$ are \emph{i.i.d.} across time and independent across users, and there are a finite number of possible channel states due to digital quantization. The state space of the imperfect CSI $\hat{\textbf{h}}(t)$ is expressed as $\hat{\mathcal{H}}\triangleq\{\hat{\textbf{h}}_1,\cdots,\hat{\textbf{h}}_E\}$ with a stationary distribution $\bm \pi=\{\pi_1,\cdots,\pi_E\}$ which is unknown to the base station.

Let $P(t)\in\mathcal {P}$ be the transmission power of the base station at slot $t$, where $\mathcal {P}\triangleq\{P_1,P_2,\cdots,P_{O}\}$ is the set of possible transmission power values.
The downlink transmissions are subject to a time-average power constraint
{\small\begin{equation}\label{eq23}
\limsup_{T\rightarrow\infty}\frac{1}{T}\sum_{t=0}^{T-1}P(t)\leq P_{av},
\end{equation}}
where $P_{av}\in\mathcal {P}$ is the maximum average transmission power. When a unicast flow is scheduled, the transmission power $P(t)$ is determined from the imperfect CSI $\hat{\textbf{h}}(t)$ to exploit multiuser diversity gain. On the other hand, the transmission power of multicast services is simply fixed as $P(t)=P_{av}$, because the base station has no knowledge of the channels for the multicast users \cite{Oyman2010}. Let $\omega(t)\triangleq(P(t),s(t))\in\Omega$ denote the power allocation and scheduling control action at slot $t$, where $\Omega\triangleq\{\omega_1,\cdots, \omega_F\}\subset\mathcal {P}\times\mathcal {S}$ is the set of possible network control actions with $\omega_m=(P^{m},s^{m})$ and $F=UO+G$.

The mutual information for the downlink channel to the $u$th unicast user is denoted by $I(h_u,P)$, and that for the $j$th user in multicast user group $g$ is denoted by $I(h_{gj},P)$.~We assume that $I(h_u,P)$ and $I(h_{gj},P)$ are bounded by 
\begin{eqnarray}\label{eq1}
0\leq I(h_u,P)\leq I_{\max}, ~~0\leq I(h_{gj},P)\leq I_{\max},
\end{eqnarray}
where the upper bound $I_{\max}$ is due to the limited dynamic range of the received signal at the RF front end.

Let $\Phi_s(t)$ denote the number of arrival bits
at the data queue of flow $s$ at slot $t$. We assume that each arrival process $\Phi_s(t)$ is an independent irreducible positive recurrent Markov chain with countable state space and satisfies the Strong Law of Large Number (SLLN): That is, with probability one
\begin{eqnarray}\label{eq:SLLN-arrival}
\lim_{t\rightarrow \infty}\frac{\sum_{\tau=0}^t\Phi_s(\tau)}{t}=\lambda_s
\end{eqnarray}
for each flow $s$, where $\lambda_s$ is the mean arrival rate of flow $s$. We let $\bm{\lambda}\triangleq(\lambda_1,\cdots,\lambda_{G+U})$ denote arrival rate vector.


%% file: sec4.tex
\section{Queueing System For Rateless Codes}\label{sec:queue_system}

\subsection{Reception Procedure of Rateless Codes}
The encoder of rateless codes can generate an unlimited number of coded packages from a given payload message. In each slot, the transmitter sends out one coded packet to the receiver. In practice, the decoding instant of rateless codes can be determined in the following way \cite{Etesami06,Draper09,Eswaran10}: The receiver records the mutual information corresponding to each coded packet, which represents the reliability of the coded packet. 
When the accumulated mutual information exceeds $M(1+\epsilon)$ bits, the receiver can decode the message with a high probability, where $M$ is the bit size of the message, $\epsilon\geq0$ is a constant called the \emph{reception overhead}. A rateless code is said to be ``good'', if $\epsilon$ is close to zero. When $M$ is large enough, $\epsilon$ is small for many practical physical-layer rateless codes, such as Raptor codes \cite{Etesami06}, Strider \cite{Gudipati2011}, and Spinal codes \cite{Balakrishnan2012}.
%
%
\subsection{Unicast Service}\label{sec:queue_unicast}
\subsubsection{Queue Updates}
Let $R_u(t)$ denote the accumulated mutual information of user $u$ for decoding the latest message. Then, the evolutions of $R_u(t)$ are given by
{\small\begin{equation}\label{eq82}
R_u(t+1)\! =\!\left\{\!\!\!\begin{array}{l l} R_u(t),&\!\!\!\!\!\!\!\!\!\!\!\!\!\!\!\!\!\!\!\!\!\!\!\!\!\!\!\!\!\!\!\!\!\!\!\!\!\!\!\!\!\!\!\!\!\!\!\!\! \textrm{if}~s(t)\neq u;\\
R_u(t)\!+\!I(h_u(t),P(t))K,\!\!\!\!\!~~&\\
&\!\!\!\!\!\!\!\!\!\!\!\!\!\!\!\!\!\!\!\!\!\!\!\!\!\!\!\!\!\!\!\!\!\!\!\!\!\!\!\!\!\!\!\!\!\!\!\!\!\textrm{if~} s(t)= u \textrm{~and}\\
&\!\!\!\!\!\!\!\!\!\!\!\!\!\!\!\!\!\!\!\!\!\!\!\!\!\!\!\!\!\!\!\!\!\!\!\!\!\!\!\!\!\!\!\!\!\!\!\!\!~R_u(t)\!+\! I(h_u(t),P(t))K\!<\!M_u(1\!+\!\epsilon);\\
0, &\!\!\!\!\!\!\!\!\!\!\!\!\!\!\!\!\!\!\!\!\!\!\!\!\!\!\!\!\!\!\!\!\!\!\!\!\!\!\!\!\!\!\!\!\!\!\!\!\! \textrm{otherwise},
\end{array}\right.
\end{equation}}
where $K$ is the number of channel symbols in each slot and $M_u$ is the bit size of each message for user $u$. Hence, $R_u(t)$ increases when user $u$ is scheduled, and is reset to zero when user $u$ can decode the message with a high probability.

Since channel decoding and queue update occur intermittently, we define a \emph{queue update variable $a(t)$}: if user $u$ can decode the message, i.e., if $s(t)=u$ and $R_u(t)+
I(h_u(t),P(t))K\geq M_u(1+\epsilon)$, it reports an ACK such that $a(t)=u$; otherwise, $a(t)\neq u$.
Let $Q_u$ be the data queue of unicast user $u$ at the transmitter. The evolutions of $Q_u$ are
\begin{equation}\label{eq85}
Q_u(t+1)=(Q_u(t)-M_u1_{\{a(t) =u \}})^++\Phi_u(t),
\end{equation}
where $(\cdot)^+=\max\{\cdot,0\}$ and $1_{\{A\}}$ is the indicator function of event $A$.
\subsubsection{Code Length}
Let $n_u(t)$ denote the code index of user $u$ at slot $t$, which evolves as
\begin{equation}\label{eq:22}
n_u(t+1) = n_u(t) +1_{\{a(t) =u\}}.
\end{equation}
Define $t_{n,u}\triangleq\min\{t\geq0:n_u(t)=n\}$ as the first time slot such that $n_u(t)=n$.
According to the decoding rule of rateless codes, the code length (i.e., the number of coded packets or transmission time slots) for the $n$th rateless code of unicast user $u$ is (see also \cite[Eq. (11)]{Draper09})
{\small\begin{eqnarray}\label{eq83}
&&\!\!\!\!\!\!\!\!\!\!\!\!\!\!L_u(n)=\! \min\!\left\{\sum_{t=t_{n,u}}^{t_{n,u}+l-1}1_{\{s(t)=u\}}\right.: \nonumber\\
&&\!\!\!\!\!\!\!\!\!\!\!\!\!\!~~\left.\!\sum_{t=t_{n,u}}^{t_{n,u}+l-1} \!\!\! 1_{\{s(t)=u\}} I(h_u(t),P(t))K\!\geq\! M_u(1+\epsilon),~l\geq1\!\!\right\}\!,~
\end{eqnarray}}
where $\sum_{t=t_1}^{t_2}1_{\{s(t)=u\}}$ represents the number of serving slots of user $u$ between slots $t_1$ and $t_2$. Note that the part of mutual information overshooting the message size, i.e. $M_u$ bits, results in a small rate loss from the empirical channel capacity.
%
%
We omit $\epsilon$ in the rest of the paper for notational simplicity.~Nevertheless,
one can always divide $I(h_u,P)$ by $(1+\epsilon)$ to derive the result
for non-zero $\epsilon$.

\subsection{Multicast Service}\label{sec:multicast_queue}
\subsubsection{Queue Updates}
Let $R_{gj}(t)$ denote the accumulated mutual information that the $j$th user in multicast group $g$ has collected for decoding the latest message. The evolutions of $R_{gj}(t)$ are determined by
{\small\begin{equation}\label{eq182}
R_{gj}(t+1)\! =\!\left\{\!\!\!\begin{array}{l l} R_{gj}(t),&\!\!\!\!\!\!\!\!\!\!\!\!\!\!\!\!\!\!\!\!\!\!\!\!\!\!\!\!\!\!\!\!\!\!\!\!\!\!\!\!\!\!\!\!\! \textrm{if}~s(t)\neq g;\\
R_{gj}(t)\!+\!I(h_{gj}(t),P_{av})K,\!\!\!\!\!\\
&\!\!\!\!\!\!\!\!\!\!\!\!\!\!\!\!\!\!\!\!\!\!\!\!\!\!\!\!\!\!\!\!\!\!\!\!\!\!\!\!\!\!\!\!\!\textrm{if}~s(t)= g~\textrm{and~} \exists~l\in\{1,\cdots,J(g)\},\\
&\!\!\!\!\!\!\!\!\!\!\!\!\!\!\!\!\!\!\!\!\!\!\!\!\!\!\!\!\!\!\!\!\!\!\!\!\!\!\!\!\!\!\!\!\!~~~R_{gl}(t)\!+\! I(h_{gl}(t),P_{av})K\!<\!M_{g};\\
0, & \!\!\!\!\!\!\!\!\!\!\!\!\!\!\!\!\!\!\!\!\!\!\!\!\!\!\!\!\!\!\!\!\!\!\!\!\!\!\!\!\!\!\!\!\!\textrm{otherwise},
\end{array}\right.\end{equation}}
where $M_{g}$ is the bit size of each message for group $g$.
Therefore, $R_{gj}(t)$ increases when the multicast group $g$ is scheduled, and is reset to zero when all the $J(g)$ multicast users can decode the message.

The queue update variable $a(t)$ for the multicast flows is determined as follows: Each multicast user reports an ACK to the base station when it has collected enough mutual information to decode the message. If all the $J(g)$ multicast users in group $g$ have sent back their ACKs by the end of slot $t$, i.e., $s(t)=g$ and $R_{gj}(t)+$ $
I(h_{gj}(t),P_{av})K\geq M_g$ for all $j\in\{1,\cdots,J(g)\}$, 
then $a(t)=g$; otherwise, $a(t)\neq g$.~Let $Q_g$ be the data queue of multicast user group $g$. The evolutions of $Q_g$ are given by
\begin{equation}\label{eq185}
Q_g(t+1)=(Q_g(t)-M_g1_{\{a(t) =g \}})^++\Phi_g(t).
\end{equation}
\subsubsection{Code Length}
Let $n_g(t)$ be the code index of multicast user group $g$, which evolves as
\begin{equation}\label{eq:20}
n_g(t+1) = n_g(t) +1_{\{a(t) =g \}}.
\end{equation}
Define $t_{n,g}\triangleq\min\{t\geq0:n_g(t)=n\}$ as the first time slot such that $n_g(t)=n$.
The code length for the $n$th rateless code of multicast user group $g$ is determined as
{\small\begin{eqnarray}\label{eq:blocklength-multicast}
&&\!\!\!\!\!\!\!\!\!\!\!\!\!\!L_g(n)=\max_{j\in\{1,\cdots,J(g)\}} L_{gj}(n),
\end{eqnarray}}
where $L_{gj}(n)$ is the number of coded packets for the $j$th multicast user in group $g$ to decode the $n$th message, i.e.,
{\small\begin{eqnarray}\label{eq:1}
&&\!\!\!\!\!\!\!\!\!\!\!\!\!\!L_{gj}(n)=\! \min\!\left\{\sum_{t=t_{n,g}}^{t_{n,g}+l-1}1_{\{s(t)=g\}}\right.: \nonumber\\
&&\!\!\!\!\!\!\!\!\!\!\!\!\!\!~~~~\left.\sum_{t=t_{n,g}}^{t_{n,g}+l-1} \!\! 1_{\{s(t)=g\}} I(h_{gj}(t),P_{av})K\geq M_g,~l\geq1\right\}\!.~~~~
\end{eqnarray}}

Since $h_{gj}(t)$ and $I(h_{gj}(t),P_{av})$ are \emph{i.i.d.} over time, they are also \emph{i.i.d.} in the slots when the $g$th group is scheduled. Furthermore, $L_g(n)$ is a stopping time that the accumulated mutual information $R_{gj}(t)$ exceeds $M_g$ bits for all the $J(g)$ multicast users. By using the property of stopping times \cite{Durrettbook10}, $L_g(n)$ are \emph{i.i.d.} for different code indices $n$ and satisfy SLLN:
{\small\begin{eqnarray}\label{eq:blocklength-SLLN1}
\lim_{N\rightarrow \infty}\frac{\sum_{n=1}^{N}L_g(n)}{N}\triangleq\overline{L}_g,
\end{eqnarray}}
where $\overline{L}_g$ is the average code length of the $g$th multicast user group that is irrelative to the network control algorithm. The multicast capacity of group $g$ is $M_g/\overline{L}_g$ bits/slot \cite{Cogill2011}.

%% file: sec5.tex
\section{Network Control for Rateless Codes} \label{sec:algorithm}
In this section, we propose a scheduling and power allocation scheme for unicast and multicast services with rateless codes. We show that this scheme achieves a near-optimal throughput region.

\subsection{Network Control Algorithm}
Let us define a virtual queue $Z$ for the time-average power constraint \eqref{eq23}, which evolves as
\begin{eqnarray}\label{eq86}
&&Z(t+1) = (Z(t)-P_{av})^++ P(t).\label{eq71}
\end{eqnarray}

We now propose a scheduling and power allocation algorithm, which only requires imperfect CSI and infrequent ACK feedback to control the wireless transmissions:\\
\underline{\textbf{Network Control For Rateless Codes (NC-RC)}}
\begin{itemize}
\item {\bf Scheduling:}
The scheduling decision at slot $t$ is
{\small\begin{eqnarray}
\label{eq96}
&&\!\!\!\!\!\!\!\!\!\!\!\!s(t) = \argmax_{s\in\mathcal {S}} {Q}_s(t) I_s(t)-Z(t)P_s(t),~~~~~~
\end{eqnarray}}
where $I_s(t)$ and $P_s(t)$ are given by
{\small\begin{eqnarray}\label{eq-rate}
&&\!\!\!\!\!\!\!\!\!\!\!\!\!I_u(t)=\mathbb{E}\{I(h_u(t),P_u(t))K|\hat{h}_u(t)\}\frac{M_u}{M_u\!+\!I_{\max}K},\nonumber\\
&&\!\!\!\!\!\!\!\!\!\!\!\!\!I_g(t)=\left\{\begin{array}{l l}\frac{n_g(t)M_g}{\sum_{n=1}^{n_g(t)-1}L_g(n)}& \textrm{,~if~} n_g(t)>1;\\
I_{\max}K& \textrm{,~if~}n_g(t)=1,\end{array}\right.\nonumber\\
&&\!\!\!\!\!\!\!\!\!\!\!\!\!P_u(t)=\argmax\limits_{P\in\mathcal {P}}Q_u(t)\mathbb{E}\{I(h_u(t),P)K|\hat{h}_u(t)\}\nonumber\\
&&~~~~~~~~~~~~~~\times\frac{M_u}{M_u+I_{\max}K}-Z(t)P,\nonumber\\
&&\!\!\!\!\!\!\!\!\!\!\!\!\!P_g(t)=P_{av},
\end{eqnarray}}
$\!\!\!$for $u\in\{1,\cdots,U\}$ and $g\in\{U+1,\cdots,U+G\}$, and $\mathbb{E}\{X|Y\}$ represents the conditional expectation of  $X$ for a given value of $Y$.

\item {\bf Power allocation:}
The transmission power at slot $t$ is
{\small\begin{eqnarray} \label{eq:power}
P(t)=P_{s(t)}(t),
\end{eqnarray}}
where $P_u(t)$ and $P_g(t)$ are defined in \eqref{eq-rate}.

\item {\bf Queue update:}
Update the queues $R_u(t)$,
$Q_u(t)$, $R_{gj}(t)$,
$Q_g(t)$, and $Z(t)$
according to \eqref{eq82}, \eqref{eq85}, \eqref{eq182}, \eqref{eq185}, and \eqref{eq86}, respectively.
\end{itemize}

In this algorithm, the unicast service rate $I_u(t)$ is the ergodic capacity of user $u$ with a rate loss factor of $M_u/(M_u+I_{\max}K)$. This rate loss is caused by the part of mutual information overshooting the message size $M_u$. The multicast service rate $I_g(t)$ is attained by using the ACKs to track the empirical average throughput of historical transmissions. In particular, $I_g(t)$ converges to the multicast capacity $M_g/\overline{L}_g$ as $t$ grows.

\subsection{Throughput Region}\label{sec:throughput1}
As in \cite{Andrews04}, a queueing network is said to be \emph{stable}, if the underlying Markov chain is positive recurrent.
A stable throughput region of the proposed network control algorithm NC-RC is stated as follows:
\begin{Theorem}\label{prop1}
The network is stable under NC-RC for any arrival rate vector $\bm{\lambda}$ strictly inside $\Lambda$,
where $\Lambda$ is given by
{\small\begin{eqnarray}\label{eq:region}
&&\!\!\!\!\!\!\!\!\!\!\!\!\!\!\!\Lambda=\bigg\{\bm{\lambda}\bigg| \textrm{There exist } \alpha_{mi}\geq0,~\textrm{such that} \nonumber\\  &&\!\!\!\!\!\!\!\!\!\!\!0\leq \lambda_u\leq\!\!\!\!\sum_{m:s^m=u}\!\sum_{i=1}^E\mathbb{E}\{I(h_u,P^m)K|\hat{\mathbf{h}}_i\}\frac{M_u}{M_u+I_{\max}K}\alpha_{mi}\pi_i,\nonumber\\
&&\!\!\!\!\!\!\!\!\!\!\!0\leq\lambda_g\leq\frac{M_g}{\overline{L}_g}\sum_{m:s^m=g}\!\sum_{i=1}^E \alpha_{mi}\pi_i ,\nonumber\\ &&\!\!\!\!\!\!\!\!\!\!\!0\leq\sum_{m=1}^F\!\sum_{i=1}^EP^m\alpha_{mi}\pi_i\leq P_{av},~\sum_{m=1}^F\alpha_{mi}=1,\forall~i \bigg\},~~
\end{eqnarray}}
$\!\!\!$where $\overline{L}_g$ is defined in \eqref{eq:blocklength-SLLN1}, $P^m$ and $s^m$ are the power allocation and scheduling decisions associated with action $\omega_m$, and $\pi_i$ is the stationary probability of the channel state $\hat{\mathbf{h}}_i$, i.e., $\pi_i=\Pr\{\hat{\mathbf{h}}(t)=\hat{\mathbf{h}}_i\}$.
\end{Theorem}

In \eqref{eq:region}, the unicast throughput $\lambda_u$ is upper bounded by the ergodic capacity of user $u$ multiplied with a rate loss factor $M_u/(M_u+I_{\max}K)$ and a time-sharing variable $\alpha_{mi}$ for selecting action $\omega_m$ when the imperfect CSI is $\hat{\mathbf{h}}(t)=\hat{\mathbf{h}}_i$. The multicast throughput $\lambda_g$ is upper bounded by the multicast capacity of group $g$ multiplied with a time-sharing variable $\sum_{i=1}^E \alpha_{mi}\pi_i$ for scheduling group $g$. Therefore, the achievable throughput region $\Lambda$ is as least $\min_{u}\{M_u/(M_u+I_{\max}K)\}$ of the optimal throughput region, without incurring excessive transmission delay. Note that the NC-RC scheme achieves $\Lambda$ with no prior knowledge of the channel distribution $\pi_i$.


The key difficulty in proving Theorem \ref{prop1} is that the temporal correlation in the transmission procedure of rateless codes. Owing to this, traditional \emph{slot-level} Lyapunov drift techniques \cite{NetworkResourceAllocationBook06} are not sufficient to prove Theorem \ref{prop1}. This time-correlation issue was resolved in \cite{Sun2013} for unicast service by constructing a capacity-achieving scheme with \emph{i.i.d.} transmission rates across time slots and evaluating the throughput difference between the two schemes. However, this method cannot be used to analyze the multicast throughput of our NC-RC algorithm, because the transmission rate $I_g(t)$ in \eqref{eq-rate} is non-\emph{i.i.d.} across time.
In this paper, we utilize fluid limit techniques
\ifreport
\cite{Chen95fluidapproximations,Dai95,Dai95-2,Rybko92,Stolyar95,Andrews04}
\fi
to prove Theorem \ref{prop1}. Since the fluid limit functions are deterministic, the service rates in the fluid limits are deterministic without time-correlation. By this, the time-correlation issue is resolved.

The details of the proof are provided in Appendix \ref{sec:analysis}. In the following we provide an outline of the proof. We first show  that the queueing system can be described as a Markov chain after adding some extra state variables, and then establish the fluid limits of this Markov chain. According to \eqref{eq83}, the accumulated mutual information of user $u$ is smaller than $M_u$ before the last coded packet is received, that is
\begin{eqnarray}\label{eq:proof311}
\sum_{t=t_{n,u}}^{t_{n,u}+l-1} \!\! 1_{\{s(t)=u\}} I(h_u(t),P(t))K< M_u,
\end{eqnarray}
where $l$ satisfies $\sum_{t=t_{n,u}}^{t_{n,u}+l-1}1_{\{s(t)=u\}}=L_u(n)-1$. Using this, we can obtain a lower bound of the queue service rate in the fluid limits (Lemma \ref{lem3}). Then, by using a Lyapunov drift techniques for the fluid limits and the stability criteria stated in Lemma \ref{lem4} for discrete-time countable Markov chains, we can establish the stability of the original queue system when the arrival rate vector $\bm{\lambda}$ is strictly inside $\Lambda$.

%% file: sec7.tex
\section{Multicast Throughput Enhancement by Unicast Retransmissions}\label{sec:uni&multi}

Let us consider the procedure of transmitting a rateless code to a group of multicast users. At the beginning, all the multicast users in the group can attain useful mutual information and the multicast gain is high. However, as more and more users have decoded the message and stopped receiving packets, the multicast gain decreases. In the end, only one or several users with poor channel quality are still receiving packets and the multicast gain becomes quite small, which significantly degrades the throughput of multicast service. In this section, we analyze the stability region of a combined delivery strategy, which first delivers the message to most users in a multicast session, and at certain time switches to unicast retransmissions \cite{MBMS_LTEA2012} to send additional coded packets to the users with poor channel conditions using higher transmission power.
\subsection{User Partition}
Suppose that the base station needs to deliver messages to the users within the set $\mathcal {L}_g$ in the multicast session, where $\mathcal {L}_g\subset\{1,\cdots,J(g)\}$. In order to improve the throughput of the multicast session, the users within $\mathcal {L}_g$ should have better channel quality than the other users in the $g$th group. The base station can attain the channel quality of each multicast user by evaluating its average throughput. In particular, the average throughput of user $j$ in the $g$th multicast group is
{\small\begin{eqnarray}
\overline{I}_{gj}\triangleq \lim_{N\rightarrow \infty } \frac{N M_g}{\sum_{n=1}^NL_{gj}(n)},
\end{eqnarray}}
where $L_{gj}(n)$ is defined in \eqref{eq:1}. Without loss of generality, we assume that the multicast users are sorted in the descending order of their average throughput, i.e., $\overline{I}_{g1}\geq \overline{I}_{g2}\geq\cdots\geq \overline{I}_{gJ(g)}$. Then, the set $\mathcal {L}_g$ can be expressed as $\mathcal {L}_g=\{1,\cdots, l(g)\}$ with $l(g)\leq J(g)$.
In practice, the number of multicast users $l(g)$ within $\mathcal {L}_g$ is attained from system requirements and user service experience. For example, the LTE-Advanced MBMS standards require to cover a percentage (e.g., $95\%$) of the users by multicast delivery \cite{MBMS_LTEA2012}.
\subsection{Queueing System}
Similar to \eqref{eq:blocklength-multicast}, the service duration for the $n$th multicast session of the $g$th group is
\begin{eqnarray}
L_{g}(n,l(g))=\max_{j\in\{1,\cdots,l(g)\}} L_{gj}(n).
\end{eqnarray}
Recall that $I(h_{gj}(t),P_{av})$ is \emph{i.i.d.} in the slots when the $g$th group is scheduled.
Since $L_{g}(n,l(g))$ is a stopping time that the multicast users within $\mathcal {L}_g$ have accumulated enough mutual information to decode the $n$th message, one can attain that $L_{g}(n,l(g))$ is \emph{i.i.d.} for different $n$ and satisfies SLLN:
\begin{eqnarray}\label{eq:blocklength-SLLN}
\lim_{N\rightarrow \infty}\frac{\sum_{n=1}^{N}L_{g}(n,l(g))}{N}\triangleq\overline{L}_g(l(g)),
\end{eqnarray}
where $\overline{L}_g(l(g))$ is the average service duration that is irrelative to the network control actions.
Therefore, the throughput of the multicast session is given by $M_g/\overline{L}_g(l(g))$ bits/slot.

The multicast users $j\in\{l(g)+1,\cdots, J(g)\}$ in the $g$th group require extra unicast sessions to improve the throughput performance.
Let $v(g,j)$ denote the extra unicast flow for the $j$th multicast user in the $g$th group, where $j\in\{l(g)+1,\cdots, J(g)\}$. We use $R^\ast_{gj}(n)\in[0,M_g]$ to represent the amount of mutual information that user $j$ of the $g$th group has collected in the $n$th multicast session. Then, the left $M_v(n)\triangleq M_g-R^\ast_{gj}(n)$ bits of mutual information needs to be retrieved in latter unicast file repair sessions. Note that $R^\ast_{gj}(n)$ and $M_v(n)$ are both \emph{i.i.d.} over $n$. Let us define $\eta_{gj}$ as the ratio of mutual information obtained in multicast sessions, i.e.,
\begin{eqnarray}\label{eq:8}
\eta_{gj}\triangleq\lim_{N\rightarrow \infty } \frac{\sum\limits_{n=1}^NR^\ast_{gj}(n)}{NM_g}.
\end{eqnarray}

In the unicast retransmission sessions, each receiver sends its imperfect CSI to the base station like other unicast users.
We define $s(t)=v$ as the event that the $v$th unicast retransmission flow is scheduled, and $a(t)=v$ as the event that a message is decoded by the receiver of the $v$th unicast retransmission flow.
The base station maintains an extra queue $Q_{v}$ for the $v$th unicast retransmission flow.
The evolutions of $Q_{v}$ are given by
\begin{eqnarray}\label{eq:44}
&&\!\!\!\!\!\!\!\!\!\!\!\!Q_{v}(t\!+\!1)\!=\!\big[Q_{v}(t)\!-\! M_v(n_v(t))1_{\{a(t)=v\}}\nonumber\\
&&~~~~~~~~~~~~~~-R^\ast_{gj}(n_g(t))1_{\{a(t)=g\}}\big]^++ \Phi_g(t),
\end{eqnarray}
where $v=v(g,j)$, $n_g(t)$ is defined in \eqref{eq:20}, and $n_v(t)$ is index of the latest message in the $v$th unicast flow similar to \eqref{eq:22}.



\begin{figure}
\centering
\includegraphics[width=2.2in]{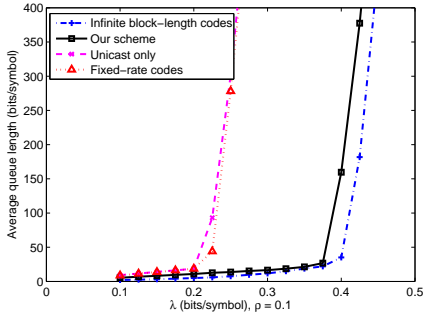}
\vspace{-0.1cm}
\caption{Simulation results of average queue length versus the traffic load $\lambda$ for $\rho = 0.1$ and $SNR=10$dB.} \label{fig1}
\vspace{-0.2cm}
\end{figure}

\subsection{Throughput Region}
We use the NC-RC algorithm to determine the scheduling and power allocation actions for this combined delivery strategy, where $J(g)$ is replaced by $l(g)$, $Q_{v}(t)$ is updated according to \eqref{eq:44}, and $I_v(t)$ is in the same form of $I_u(t)$.

A stable throughput region of the combined delivery strategy is stated as follows:
\begin{Theorem}\label{thm2}
The network with the combined delivery strategy is stable under NC-RC for any arrival rate vector $\bm{\lambda}$ strictly inside $\Lambda_2$, where $\Lambda_2$ is given by
{\small\begin{eqnarray}\label{eq:region1}
&&\!\!\!\!\!\!\!\!\!\!\!\!\Lambda_2=\bigg\{\bm{\lambda}\bigg| \textrm{There exist } \alpha_{mi}\geq0,~\textrm{such that} \nonumber\\  &&\!\!\!\!\!\!\!\!\!\!\!0\leq \lambda_u\leq\!\!\!\!\sum_{m:s^m=u}\!\sum_{i=1}^E\mathbb{E}\{I(h_u,P^m)K|\hat{\mathbf{h}}_i\}\frac{M_u}{M_u+I_{\max}K}\alpha_{mi}\pi_i,\nonumber\\
&&\!\!\!\!\!\!\!\!\!\!\!0\leq\lambda_g\leq\frac{M_g}{\overline{L}_g(l(g))}\sum_{m:s^m=g}\!\sum_{i=1}^E\alpha_{mi}\pi_i ,\nonumber\\
&&\!\!\!\!\!\!\!\!\!\!\!0\leq\lambda_g\leq\!\!\!\sum_{m:s^m=v(g,j)}\!\sum_{i=1}^E\mathbb{E}\{I(h_{gj},P^m)K|\hat{\mathbf{h}}_i\}\frac{(1-\eta_{gj})M_g}{(1\!-\!\eta_{gj})M_g\!+\!I_{\max}K}\nonumber\\
&&\!\!\!\!\!\!\!\!\!\!\!~~~~\times\alpha_{mi}\pi_i+\frac{\eta_{gj}M_g}{\overline{L}_g(l(g))}\!\!\sum_{m:s^m=g}\!\sum_{i=1}^E\alpha_{mi}\pi_i,j=l(g)+1,\cdots,J(g),\nonumber\\
&&\!\!\!\!\!\!\!\!\!\!\!0\leq\sum_{m=1}^F\!\sum_{i=1}^EP^m\alpha_{mi}\pi_i\leq P_{av},~\sum_{m=1}^{M}\alpha_{mi}=1,\forall~i \bigg\},~~
\end{eqnarray}}
$\!\!$where $\overline{L}_g(l(g))$ and $\eta_{gj}$ are defined in \eqref{eq:blocklength-SLLN}, and \eqref{eq:8}, respectively, $P^m$ and $s^m$ are the power allocation and scheduling decisions associated with action $\omega_m$, $\hat{\mathbf{h}}_i$ is the imperfect CSI state for unicast users and the multicast users with file repair.
\end{Theorem}
Note that $\lambda_g$ is upper bounded by not only the multicast throughput of the users within the set $\{1,\cdots,J(g)\}$, but also the total throughput of multicast and unicast services for each user within the set $\{l(g)+1,\cdots,J(g)\}$.
\ifreport
\begin{proof}
See Appendix \ref{sec:analysis2}.
\end{proof}
\fi

\begin{figure}
\centering
\includegraphics[width=2.2in]{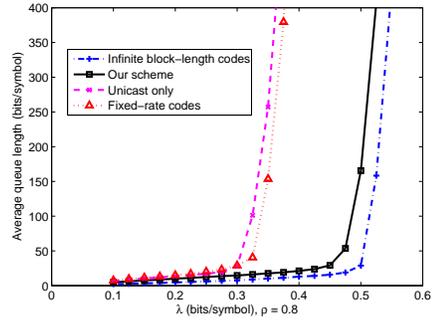}
\vspace{-0.1cm}
\caption{Simulation results of average queue length versus the traffic load $\lambda$ for $\rho = 0.8$ and $SNR=10$dB.} \label{fig2}
\vspace{-0.2cm}
\end{figure}

%% file: sec8.tex
\section{Numerical Evaluation}\label{sec:simulation}
We evaluate the performance of NC-RC algorithm based on simulations. Consider a downlink network with $U = 5$ unicast users and $G=2$ multicast user groups. Each multicast group contains $J(g)=4$ users. The downlink wireless channels are modeled by
Rayleigh fading. We adopt an additive channel uncertainty model, e.g., \cite{Goldsmith2003}, in which the imperfect CSI $\hat{h}_u(t)$ is determined by
\begin{eqnarray}
&&\hat{h}_u(t) = \sqrt{\rho} h_u(t) + \sqrt{1-\rho}\hat{n}_u(t),
\end{eqnarray}
where $n_u(t)$ is circular-symmetric complex Gaussian variable with zero mean and the same variance of $h_u(t)$, and $\rho$ is the correlation coefficient which represents the accuracy of the imperfect CSI. {All the channel states are quantized into 16-bit discrete values due to analog-to-digital conversion.} The mutual information
is expressed as $I(h,P)=\max\{\log_2(1+|h|^2P),5\}$, where the additional upper bound $I_{\max}=5$ bits/symbol is due to the limited
dynamic range of practical RF receivers. The message size of each rateless code is $M_s=40$ bits, and the number of channel symbols in each slot is normalized as $K=1$. The average SNR of a wireless channel is determined as $SNR = E\{|h(t)|^2\}P_{av}$. The traffic loads of all the flows are chosen to be the same, i.e., $\lambda_u=\lambda_g=\lambda$.

Three reference strategies are considered for performance
comparison: The first strategy uses infinite block-length channel codes, which achieves an
outer bound of the stability region $\Lambda$ in \eqref{eq:region}. However, this strategy is not practical since it needs prior knowledge of the channel distribution $\pi_i$ and has an infinite transmission delay. The second one uses
fixed-rate codes on the physical layer and rateless codes on the application layer \cite{Calabuig2013}, where the packets not successfully decoded in the physical layer are discarded from application layer message decoding.
The third strategy utilizes unicast sessions based on physical-layer rateless codes to transmit service data to all the users \cite{Sun2013}, where all the physical-layer packets contribute to message decoding by information accumulation.
Near-optimal scheduling and power allocation schemes are designed for these strategies.
{As we have mentioned in Section \ref{sec:queue_unicast}, the throughput performance of practical rateless codes can be obtained by incorporating their reception overhead $\epsilon$ and divide the realized throughput by $(1+\epsilon)$.}

\begin{figure}
\centering
\includegraphics[width=2.53in]{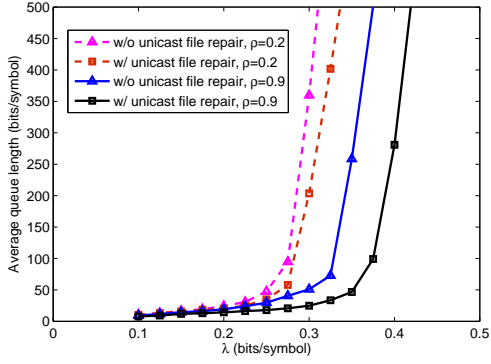}
\vspace{-0.1cm}
\caption{Simulation results of average queue length versus the traffic load $\lambda$ for unicast file repair in a heterogeneous SNR scenario.} \label{fig3}
\vspace{-0.2cm}
\end{figure}

Figures \ref{fig1} and \ref{fig2} illustrate the comparison results of average queue length versus the traffic load $\lambda$ for $\rho = 0.1$ and $\rho = 0.8$, respectively, where the users have the same average SNR of $10$ dB. As expected, the performance of all four schemes improves with increasing $\rho$. The performance of our scheme with rateless codes is quite close to that achieved by using infinite block-size codes. In fact, our scheme can achieve at least $M_u/(M_u+I_{\max}K)=8/9$ of the optimal throughput region. The performance of physical-layer fixed-rate codes is much worse than that of our scheme. {The reason is that the maximum achievable goodput $\max_RR\Pr\{I(h_u,P)\geq R|\hat{h}_u\}$ could be much smaller than the ergodic capacity $\mathbb{E}\{I(h_u,P)|\hat{h}_u\}$ for unicast service. Similar rate loss also exists in multicast service.} We note that although application-layer rateless codes can achieve multicast gain and reduce feedback overhead, they cannot recover the rate loss caused by fixed-rate codes in the physical layer. The performance of unicast only scheme is bad, because it has not exploited the multicast gain.

Figure \ref{fig3} provides the results of our schemes with and without unicast file repair. The average SNRs of the 5 unicast users are $[12,10,8,6,4]$ dB, the average SNRs of the 2 groups of multicast users are $[12,9,6,3; 12,9,6,3]$ dB, and $l(g)=3$ for both groups. When $\rho=0.2$, unicast file repair can achieve a higher throughput, because the base station uses more power to serve the last multicast user with poor channel quality. When $\rho=0.9$, the throughput benefit of unicast file repair becomes larger, because the opportunistic gain provided by power allocation increases with $\rho$.

%% file: sec9_conclusion.tex
\section{Conclusion}
We have investigated the management of network resources under imperfect CSI and infrequent ACK feedback.
To that end, a scheduling and power allocation strategy has been developed for downlink networks with both multicast and unicast services using rateless codes. Our strategy can simultaneously realize the benefits of multiuser diversity gain, multicast gain, and achieve robustness against channel uncertainty. Our simulation results suggest that our strategy can significantly improve the network throughput, compared to schemes using fixed-rate codes or relying on unicast communications to serve all the users.

%% file: appendex0.tex
\section{Proof of Theorem \ref{prop1}} \label{sec:analysis}
\subsection{A Markov Chain of the Queueing System}
Let us define $T_s(t)\triangleq\sum_{t=t_{n,s}}^{t}1_{\{s(t)=s\}}$ as the number of slots that has been used to send the latest message of flow $s$. The evolutions of $T_s(t)$ is given by
\begin{eqnarray}
T_s(t+1)=
\left\{
\begin{array}{l l}
T_s(t)+1_{\{s(t)=s\}}&,~\textrm{if~} a(t)\neq s;\\
0&,~\textrm{otherwise}.
\end{array}\right.
\end{eqnarray}
We further define $\Upsilon_g(t)\triangleq\sum_{n=1}^{n_g(t)-1}L_g(n)$ as the accumulated block-length of the rateless codes of group $g$, which evolves as
\begin{eqnarray}
\Upsilon_g(t+1)=\Upsilon_g(t)+ (T_g(t)+1)1_{\{a(t)=g\}}.
\end{eqnarray}

Let $\mathcal {X}(t)\triangleq(R_u(t),R_{gj}(t),Q_s(t),Z(t),\frac{n_s(t)}{t+1},\frac{T_s(t)}{t+1},$ $\frac{\Upsilon_g(t)}{t+1})$ denote the system state.
\ifreport
Note that we use $\frac{n_s(t)}{t+1}$ instead of $\frac{n_s(t)}{t}$, so that it is well-defined when $t=0$.
Since there are finite number of possible CSI values and control actions, the state spaces of $R_u(t)$, $R_{gj}(t)$, $Q_s(t)$, $Z(t)$, $n_s(t)$, $T_s(t)$, $\Upsilon_g(t)$ are all countable. Moreover, the ratios $\frac{n_s(t)}{t+1}$, $\frac{T_s(t)}{t+1}$, and $\frac{\Upsilon_g(t)}{t+1}$ also have countable state space.
\fi
In the NC-RC algorithm, the network control action $\{s(t),P(t)\}$ at slot $t$ is determined from $\mathcal {X}(t)$ and $\hat{\textbf{h}}(t)$. The service update variable $a(t)$ is determined from $\mathcal {X}(t)$, $\{s(t),P(t)\}$, and the channel states $\{h_u(t),h_{gj}(t)\}$. In addition, one can observe that the system state $\mathcal {X}(t+1)$ can be derived from $\mathcal {X}(t)$, $\{s(t),P(t)\}$, and $a(t)$. Therefore, the process $\mathcal {X}=(\mathcal {X}(t) ,t\geq0)$ is a discrete-time countable Markov chain.

\subsection{Fluid Limits}
We now establish the fluid limit model of $\mathcal {X}$. 
{Let us define the norm of $\mathcal {X}(t)$ as $||\mathcal {X}(t)||\triangleq\sum_{s=1}^{U+G} |Q_s(t)|+|Z(t)|+|R_u(t)|+|R_{gj}(t)|+|\frac{n_s(t)}{t+1}|+|\frac{T_s(t)}{t+1}|+|\frac{\Upsilon_g(t)}{t+1}|$.}%
~Let $\mathcal {X}^{(x)}$ denote a process $\mathcal {X}$ with an initial state satisfying
\begin{eqnarray}\label{eq:init_con}
||\mathcal {X}^{(x)}(0)||=x.
\end{eqnarray}

Let $A_s(t)\triangleq\sum_{\tau=0}^t\Phi_s(\tau)$ and $D_s(t)$ denote the accumulated arrival and departure bits at queue $Q_s$ up to slot $t$, respectively.~We adopt the convention that $A_s(0)=0$ and $D_s(0)=0$.
Let $\Psi_s(t)\triangleq M_{s}1_{\{a(t) =s \}}$ denote the service rate of $Q_s$ at slot $t$.~Since the queue can be empty when it is scheduled, we have $D_s(t)-D_s(t-1)\leq \Psi_s(t)$.
The queue length $Q_s$ can be described in an alternative form
\begin{eqnarray}\label{eq:queue_evo1}
Q_s(t)=Q_s(0)+A_s(t)-D_s(t).
\end{eqnarray}

Let $W(t)$ and $\Theta(t)$ denote the accumulated arrival and departure power of the virtual queue $Z$, respectively.~Therefore, we attain $W(t)\triangleq\sum_{\tau=0}^tP(\tau)$ and $\Theta(t)-\Theta(t-1)\leq P_{av}$.~Then, the virtual queue $Z$ can be also written as
\begin{eqnarray}\label{eq:eq1}
Z(t)=Z(0)+W(t)-\Theta(t).
\end{eqnarray}

Let $B_{mi}(t)$ denote the time fraction up to slot $t$ when the network control action is $\omega_m$ and the imperfect channel state is $\hat{\mathbf{h}}_i$, defined by
\begin{eqnarray}\label{eq:eq20}
B_{mi}(t)\triangleq\sum_{\tau=0}^t1_{\{\omega(\tau)=\omega_m,\hat{\textbf{h}}(\tau)=\hat{\textbf{h}}_i\}}.
\end{eqnarray}

Let us define another process $\mathcal {Y}=(\mathcal {X},A_s,D_s,$ $W,\Theta,B_{mi},\Psi_s)$, where the tuple denotes a list of processes.~Therefore, the sample path of $\mathcal {Y}^{(x)}$ uniquely define the sample path of $\mathcal {X}^{(x)}$.~ We extend the definition of $\mathcal {Y}$ to each continuous time $t\geq0$ as $\mathcal {Y}^{(x)}(t)=\mathcal {Y}^{(x)}(\lfloor t \rfloor)$.

Recall that a sequence of functions $f_n(\cdot)$ is said to converge to a function $f(\cdot)$ uniformly over compact (u.o.c) if for all $t\geq0$, $\lim_{n\rightarrow\infty}\sup_{0\leq t'\leq t}|f_n(t')-f(t')|=0$. We now consider a sequence of processes $\{\frac{1}{x_n}\mathcal {Y}^{(x_n)}(x_n\cdot)\}$ that is scaled both in time and space, and show the convergence properties of the sequences in the following lemma:
\begin{lemma} \label{lem1}
With probability one, for any sequence of the processes $\{\frac{1}{x_n}\mathcal {Y}^{(x_n)}(x_n\cdot)\}$, where $x_n$ is a sequence of positive integers with $x_n\rightarrow\infty$, there exists a subsequence $x_{n_k}$ with $x_{n_k}\rightarrow\infty$ as $k\rightarrow\infty$ such that the following u.o.c convergences hold:
{\small\begin{eqnarray}
&&\frac{1}{x_{n_k}}Q_s^{(x_{n_k})}(x_{n_k}t)\rightarrow q_s(t),\label{eq:uoc1}\\
&&\frac{1}{x_{n_k}}A_s^{(x_{n_k})}(x_{n_k}t)\rightarrow a_s(t),\label{eq:uoc5}\\
&&\frac{1}{x_{n_k}}D_s^{(x_{n_k})}(x_{n_k}t)\rightarrow d_s(t),\label{eq:uoc6}\\
&&\frac{1}{x_{n_k}}Z^{(x_{n_k})}(x_{n_k}t)\rightarrow z(t),\label{eq:uoc4}\\
&&\frac{1}{x_{n_k}}W^{(x_{n_k})}(x_{n_k}t)\rightarrow w(t),\label{eq:uoc8}\\
&&\frac{1}{x_{n_k}}\Theta^{(x_{n_k})}(x_{n_k}t)\rightarrow \theta(t),\label{eq:uoc9}\\
&&\frac{1}{x_{n_k}}B_{mi}^{(x_{n_k})}(x_{n_k}t)\rightarrow b_{mi}(t),\label{eq:uoc9}\\
&&\frac{1}{x_{n_k}}\int_0^t\Phi_s^{(x_{n_k})}(x_{n_k}\tau)d\tau\rightarrow\int_0^t\phi_s(\tau)d\tau,\label{eq:uoc10}\\
&&\frac{1}{x_{n_k}}\int_0^t \Psi_s^{(x_{n_k})}(x_{n_k}\tau) d\tau\rightarrow\int_0^t\psi_s(\tau)d\tau,\label{eq:uoc11}\\
&&\frac{1}{x_{n_k}}\int_0^t P^{(x_{n_k})}(x_{n_k}\tau) d\tau\rightarrow\int_0^tp(\tau)d\tau,\label{eq:uoc13}
\end{eqnarray}}
$\!\!$where the functions $q_s,a_s,d_s,z,w,\theta,b_{mi}$ are Lipschitz
continuous in $[0,\infty)$.
\end{lemma}
\ifreport
\begin{proof}
See Appendix \ref{App:lem1}.
\end{proof}
\fi

Since the limiting functions $q_s,a_s,d_s,z,w,\theta,b_{mi}$ are Lipschitz
continuous in $[0,\infty)$, they are absolutely continuous.~Therefore, these limiting functions are differentiable at almost all time $t\in[0,\infty)$, which we call \emph{regular time}.

\begin{lemma}\label{lem2}
Any fluid limit $(q_s,a_s,d_s,z,w,\theta,b_{mi})$ satisfies the following equations:
{\small\begin{eqnarray}
&&q_s(t)=q_s(0)+a_s(t)-d_s(t),\label{eq:fluid_eq2}\\
&&a_s(t)=\lambda_st,\label{eq:fluid_eq1}\\
&&d_s(t)\leq\int_0^t\psi_s(\tau)d\tau,\label{eq:fluid_eq4}\\
&&\frac{d}{dt}q_s(t)=\left\{\begin{array}{l l}
\lambda_s-\psi_s(t), &\textrm{if}~ q_s(t)>0;\\
(\lambda_s-\psi_s(t))^+,&\textrm{otherwise},\end{array}\right.\label{eq:fluid_eq5}\\
&&z(t)=z(0)+w(t)-\theta(t),\label{eq:fluid_eq9}\\
&&w(t)=\int_0^tp(\tau)d\tau,\label{eq:fluid_eq10}\\
&&\theta(t)\leq P_{av}t,\label{eq:fluid_eq11}\\
&&\frac{d}{dt}z(t)=\left\{\begin{array}{l l}
p(\tau)-P_{av}, &\textrm{if}~z(t)>0;\\
(p(\tau)-P_{av})^+,&\textrm{otherwise},\end{array}\right.\label{eq:fluid_eq12}\\
&& \sum_{m=1}^F  b_{mi}(t)=\pi_it.\label{eq:fluid_eq13}
\end{eqnarray}}
\end{lemma}
\ifreport
\begin{proof}
See Appendix \ref{App:lem2}.
\end{proof}
\fi

Let us define the functions
\begin{eqnarray}\label{eq:fluid_eq3}
&&c_{mi}(t)\triangleq\frac{1}{\pi_i}\frac{d}{dt}b_{mi}(t),
\end{eqnarray}
for all regular time $t\geq0$. Then, by \eqref{eq:fluid_eq13}, we attain
\begin{eqnarray}
\sum_{m=1}^F c_{mi}(t)=1,~\forall~i.
\end{eqnarray}

\begin{lemma}\label{lem3}
The fluid limit functions satisfy the properties:
{\small\begin{eqnarray}\label{eq:lem3-1}
&&\!\!\!\!\!\!\!\!\!\!\!\!\psi_u(t)\geq \sum_{m:s^m=u} \sum_{i=1}^E \mathbb{E}\{I(h_u,P^m)K|\hat{\mathbf{h}}_i\}\frac{M_uc_{mi}(t)\pi_i}{M_u+I_{\max}K},~\label{eq:lem3-1}\\
&&\!\!\!\!\!\!\!\!\!\!\!\!\psi_g(t)= \sum_{m:s^m=g} \sum_{i=1}^E \frac{M_g}{\overline{L}_g}c_{mi}(t)\pi_i,\label{eq:lem4-1}\\
&&\!\!\!\!\!\!\!\!\!\!\!\!p(t)= \sum_{m=1}^F \sum_{i=1}^E P^mc_{mi}(t)\pi_i,\label{eq:lem5-1}
\end{eqnarray}}
$\!\!$for all $u\in\{1,\cdots,U\}$ and $g\in\{U+1,\cdots,U+G\}$.
\end{lemma}

\ifreport
\begin{proof}
See Appendix \ref{App:lem3}.
\end{proof}
\fi

\subsection{Stability Analysis}\label{sec:app_sta}
The following lemma provides a stability criteria for distrete-time countable Markov chains, first obtained by Malyshev and Menshikov \cite{Malyshev79}. This stability criteria was also derived in \cite{Stolyar92} for continuous-time countable homogeneous Markov chains.
\begin{lemma}\cite[Theorem 4]{Andrews04} \label{lem4}
Suppose that there exist an $\epsilon> 0$ and a finite integer $T > 0$ such that for any
sequence of processes $\{\frac{1}{x}S^{(x)}(xT), x=1,2, \cdots\}$, we have
\begin{eqnarray}\label{eq:condition}
\limsup\limits_{x\rightarrow\infty}\mathbb{E}\left[\frac{1}{x}||S^{(x)}(xT)||\right]\leq 1-\epsilon.
\end{eqnarray}
Then, the Markov chain $S$ is stable.
\end{lemma}

\ifreport
Note that from system causality, we have $n_s(t)\leq t$ for all flow $s$. Then, we have
\begin{eqnarray}
\lim_{k\rightarrow \infty} \frac{1}{x_{n_k}} \left|\frac{1}{x_{n_k}t+1}n_s^{(x_{n_k})}(x_{n_k}t)\right| = 0
\end{eqnarray}
almost surely for all $t\geq0$. Similarly, since $T_s(t)\leq t$ and $\Upsilon_g(t)\leq t$, the fluid limits of $\frac{T_s(t)}{t+1}, \frac{\Upsilon_g(t)}{t+1}$ are also zero functions almost surely. In addition, the accumulated mutual information $R_u(t),R_{gj}(t)$ satisfies $R_u(t)\leq I_{\max}Kt$ and $R_{gj}(t)\leq I_{\max}Kt$, hence their fluild limits are zero functions almost surely. Therefore, it remains to show that the fluid limit model of the subsystem $(Q_s(t),Z(t))$ is stable in the sense of \eqref{eq:condition}.
\fi


%
%
%
%

Let us consider a quadratic Lyapunov function for the fluid limit system:
\begin{eqnarray}
L(t)=\frac{1}{2}\sum_{s=1}^{U+G}{q}_s(t)^2+z(t)^2.
\end{eqnarray}
Then, we can establish the following statement:
\begin{lemma}\label{lem5}
Consider a network under the RNC-RC algorithm for any arrival rate vector $\bm{\lambda}$ strictly inside $\Lambda$. For any $\delta_1\geq0$, there exists an $\delta_2>0$ such that the fluid limit functions satisfy the following property with probability 1: At any regular time $t$,
\begin{eqnarray}
L(t)\geq \delta_1~ \textrm{implies}~ \frac{d}{dt} L(t)\geq -\delta_2.
\end{eqnarray}
\end{lemma}
\ifreport
\begin{proof}
See Appendix \ref{App:lem5}.
\end{proof}
\fi

Lemma \ref{lem5} implies that for any $\zeta\in(0,1)$, there exists a finite $T>0$ such that $\sum_{s=1}^{U+G} |q_s(T)|+|z(T)|\leq\zeta.$ In other words, for any sequence of processes $\frac{1}{x_n}\mathcal {X}^{(x_n)}(x_nT)$, there exists a subsequence such that
\begin{eqnarray}
\lim_{k\rightarrow\infty}\frac{1}{x_{n_k}} ||\mathcal {X}^{(x_{n_k})}(x_{n_k}T)||\leq \zeta\triangleq1-\epsilon.
\end{eqnarray}
This further implies that with probability 1
\begin{eqnarray}
\limsup_{x\rightarrow\infty}\frac{1}{x} ||\mathcal {X}^{(x)}(xT)||\leq 1-\epsilon
\end{eqnarray}
holds, because there must exist a subsequence
of ${x}$ that converges to the same limit as $\limsup_{x\rightarrow\infty}\frac{1}{x} ||\mathcal {X}^{(x)}(xT)||$.

According to \eqref{eq:queue_evo1} and \eqref{eq:eq1}, we have
$||\mathcal {X}^{(x)}(xT)||\leq x + \sum_{s=1}^{U+G}A_s(xT)+W(xT)$. Hence,
\begin{eqnarray}
\mathbb{E}\left\{\frac{1}{x}||\mathcal {X}^{(x)}(xT)||\right\}\leq 1 + \sum_{s=1}^{U+G}\lambda_s T + P_{\max} T\triangleq F(T)\leq \infty.\nonumber
\end{eqnarray}
Therefore, it follows from the Dominated Convergence Theorem that
\begin{eqnarray}
&&\!\!\!\!\!\!\!~~~\limsup_{k\rightarrow\infty}\mathbb{E}\left[\frac{1}{x}||\mathcal {X}^{(x)}(xT)||\right]\nonumber\\
&&\!\!\!\!\!\!\!= \mathbb{E}\left[\limsup_{k\rightarrow\infty}\frac{1}{x}||\mathcal {X}^{(x)}(xT)||\right]\leq 1-\epsilon,
\end{eqnarray}
and thus the condition of Lemma \ref{lem4} is satisfied. This completes the proof of
Theorem \ref{prop1}.

%% file: appendex.tex
\section{Proof of Lemma \ref{lem1}}\label{App:lem1}

It follows from the SLLN \eqref{eq:SLLN-arrival} that
$\frac{1}{x_{n_k}}\int_0^t\Phi_s(x_{n_k}\tau)d\tau\rightarrow\lambda_st$.
Moreover, since $A_s(t)=\int_0^t\Phi_s(\tau)d\tau$, the convergences \eqref{eq:uoc5} and \eqref{eq:uoc10} hold, and each of the limiting functions are Lipschitz continuous.
For any given $0\leq t_1\leq t_2$, since $0\leq D_s(t)-D_s(t-1)\leq\Psi_s(t)\leq M_s$, we have that
\begin{eqnarray}\label{eq:lip1}
0\leq\frac{1}{x_{n}}\left[D_s^{(x_{n})}(x_{n}t_2)-D_s^{(x_{n})}(x_{n}t_1)\right]\leq M_s(t_2-t_1).
\end{eqnarray}
Thus, the sequence of functions $\{\frac{1}{x_{n}}D_s^{(x_{n})}(x_{n}\cdot)\}$ is uniformly bounded and uniformly equicontinuous. Consequently, by the Arzela-Ascoli Theorem, there must exist a subsequence along which \eqref{eq:uoc6} holds. Moreover,
\eqref{eq:lip1} also implies that each of the limiting functions $d_s$ is Lipschitz continuous. Using similar arguments, the convergences \eqref{eq:uoc8}-\eqref{eq:uoc9} can be shown and each of the limiting functions $w(t),\theta(t),$ and $b_{mi}$ is Lipschitz continuous.

Since the sequence $\{\frac{1}{x_{n}}Q_s^{(x_{n})}(0)\}$ are upper bounded by 1 due to \eqref{eq:init_con}, there exists a subsequence (for simpleness, assume the subsequence is denoted by $\{x_{n_k}\}$) such that $\frac{1}{x_{n_k}}Q_s^{(x_{n_k})}(0)\rightarrow q_s(0)$. Hence, convergence \eqref{eq:uoc1} simply follows from \eqref{eq:queue_evo1}. Also, each of the limiting functions $q_s(t)$ is Lipschitz continuous. Using similar arguments, we can prove the result for \eqref{eq:uoc4}.

Since $\Psi_s(t)=M_{s}1_{\{a(t) =s \}}$, we have
\begin{eqnarray}\label{eq:lip2}
\frac{1}{x_{n}}\int_{t_1}^{t_2} \Psi_s^{(x_{n})}(x_{n}\tau) d\tau\leq M_s(t_2-t_1).
\end{eqnarray}
Then by applying the Arzela-Ascoli Theorem again for \eqref{eq:lip2}, there must exist a subsequence along which \eqref{eq:uoc11} holds. Using similar arguments, we can prove \eqref{eq:uoc13}.

\section{Proof of Lemma \ref{lem2}}\label{App:lem2}
First, \eqref{eq:fluid_eq1} follows from the SLLN \eqref{eq:SLLN-arrival}. Equations \eqref{eq:fluid_eq2}, \eqref{eq:fluid_eq4},  \eqref{eq:fluid_eq9}-\eqref{eq:fluid_eq11} are satisfied from the definitions. By \eqref{eq:eq20}, we derive $\sum_{m=1}^F B_{mi}(t)= \sum_{\tau=0}^t1_{\{\hat{\textbf{h}}(\tau)=\hat{\textbf{h}}_i\}}$. Since $\hat{\textbf{h}}(t)$ is \emph{i.i.d.} over time, $1_{\{\hat{\textbf{h}}(t)=\hat{\textbf{h}}_i\}}$ is also \emph{i.i.d.}, hence \eqref{eq:fluid_eq13} follows from the SLLN. Since each of the limiting functions $q_s$ is differentiable at any regular time $t\geq0$, \eqref{eq85} and \eqref{eq185} can be rewritten as \eqref{eq:fluid_eq5}. Similarly, \eqref{eq:fluid_eq12} follows from \eqref{eq86}.

\section{Proof of Lemma \ref{lem3}}\label{App:lem3}
\subsection{Proof of Eq. \eqref{eq:lem3-1}}\label{sec:proof_lem3_1}
From \eqref{eq:uoc11}, we attain
\begin{eqnarray}\label{eq:proof2-lim4}
&&\!\!\!\!\!\!\!\!\!~~~\psi_u(t)\nonumber\\
&&\!\!\!\!\!\!\!\!\!=\frac{d}{dt}\int_{0}^t\psi_u(\tau)d\tau\nonumber\\
&&\!\!\!\!\!\!\!\!\!=\lim_{\delta\rightarrow 0}\frac{\int_{0}^{t+\delta}\psi_u(\tau)d\tau-\int_{0}^{t}\psi_u(\tau)(\tau)d\tau}{\delta }\nonumber\\
&&\!\!\!\!\!\!\!\!\!=\lim_{\delta\rightarrow 0} \lim_{k\rightarrow \infty}
\frac{1}{\delta x_{n_k}}\sum_{\tau=\lfloor tx_{n_k}\rfloor+1}^{\lfloor (t+\delta)x_{n_k}\rfloor}\Psi_u^{(x_{n_k})}(\tau)\nonumber\\
&&\!\!\!\!\!\!\!\!\!=\lim_{\delta\rightarrow 0}\lim_{k\rightarrow \infty}
\frac{1}{\delta x_{n_k}}\sum_{\tau=\lfloor tx_{n_k}\rfloor+1}^{\lfloor (t+\delta)x_{n_k}\rfloor}M_{u}1_{\{a(\tau) =u \}}.
\end{eqnarray}
According to \eqref{eq83}, the accumulated mutual information of user $u$ is smaller than $M_u$ before the last coded packet is received. This implies
\begin{eqnarray}\label{eq:proof31}
\sum_{t=t_{n,u}}^{t_{n,u}+l-1} \!\! 1_{\{s(t)=u\}} I(h_u(t),P(t))K< M_u,
\end{eqnarray}
where $l$ satisfies $\sum_{t=t_{n,u}}^{t_{n,u}+l-1}1_{\{s(t)=u\}}=L_u(n)-1$.
Taking the summation over time on both sides of \eqref{eq:proof31} yields
\begin{eqnarray}
\sum_{\tau=t_1}^{t_2} I(h_u(\tau),P(\tau))K\left(1_{\{s(\tau) =u \}}-1_{\{a(\tau) =u \}}\right)\nonumber\\
<
\sum_{\tau=t_1}^{t_2}  M_{u}1_{\{a(\tau) =u \}}+M_u\nonumber
\end{eqnarray}
for all $0\leq t_1<t_2$.~Therefore, we have
\begin{eqnarray}\label{eq:proof32}
&&~~~\sum_{\tau=t_1}^{t_2}  I(h_u(\tau),P(\tau))K1_{\{s(\tau) =u  \}}\nonumber\\
&&<\sum_{\tau=t_1}^{t_2}  \left[M_{u}+I(h_u(\tau),P(\tau))K\right]1_{\{a(\tau) =u \}}+M_u\nonumber\\
&&\!\!\!\overset{\eqref{eq1}}{<}\!\!\!\sum_{\tau=t_1}^{t_2}  (M_{u}+I_{\max}K)1_{\{a(\tau) =u \}}+M_u.\nonumber
\end{eqnarray}
Hence,
\begin{eqnarray}
&&~~~\frac{M_u}{M_{u}+I_{\max}K}\sum_{\tau=t_1}^{t_2}  I(h_u(\tau),P(\tau))K1_{\{s(\tau) =u  \}}\nonumber\\
&&\!\!\!\overset{}{<}\!\!\!\sum_{\tau=t_1}^{t_2}  M_{u}1_{\{a(\tau) =u \}}+\frac{M_u^2}{M_{u}+I_{\max}K}.\nonumber
\end{eqnarray}

By this, we can derive
\begin{eqnarray}\label{eq:proof33}
&&\!\!\!\!\!\!\!\!\!\!\!\!~~~\lim_{k\rightarrow \infty}\frac{1}{\delta x_{n_k}}\sum_{\tau=\lfloor tx_{n_k}\rfloor+1}^{\lfloor (t+\delta)x_{n_k}\rfloor}M_{u}1_{\{a(\tau) =u \}}\nonumber\\
&&\!\!\!\!\!\!\!\!\!\!\!\!\geq\lim_{k\rightarrow \infty}\frac{1}{\delta x_{n_k}}\bigg[\sum_{\tau=\lfloor tx_{n_k}\rfloor+1}^{\lfloor (t+\delta)x_{n_k}\rfloor} I(h_u(\tau),P(\tau))K1_{\{s(\tau) =u  \}}\nonumber\\
&&\!\!\!\!\!\!\!\!~~~~~~~~~~~\times\frac{M_{u}}{M_{u}+I_{\max}K}-\frac{M_u^2}{M_{u}+I_{\max}K}\bigg]\nonumber\\
&&\!\!\!\!\!\!\!\!\!\!\!\!=\lim_{k\rightarrow \infty}\frac{1}{\delta x_{n_k}}\sum_{\tau=\lfloor tx_{n_k}\rfloor+1}^{\lfloor (t+\delta)x_{n_k}\rfloor} I(h_u(\tau),P(\tau))K1_{\{s(\tau) =u  \}}\nonumber\\
&&\!\!\!\!\!\!\!\!~~~~~~~~~~~\times\frac{M_{u}}{M_{u}+I_{\max}K}\nonumber\\
&&\!\!\!\!\!\!\!\!\!\!\!\!=\lim_{k\rightarrow \infty}\frac{1}{\delta x_{n_k}}\sum_{\tau=\lfloor tx_{n_k}\rfloor+1}^{\lfloor (t+\delta)x_{n_k}\rfloor}  \sum_{m:s^m=u} \sum_{i=1}^E I(h_u(\tau),P^m)K\nonumber\\
&&\!\!\!\!\!\!\!\!~~~~~~~~~~~\times 1_{\{\omega(\tau)=\omega_m,\hat{\textbf{h}}(t)
=\hat{\textbf{h}}_i\}}\frac{M_{u}}{M_{u}+I_{\max}K}\nonumber\\
&&\!\!\!\!\!\!\!\!\!\!\!\!=\sum_{m:s^m=u} \sum_{i=1}^E \lim_{k\rightarrow \infty}
\frac{1}{\delta x_{n_k}}\sum_{\tau=\lfloor tx_{n_k}\rfloor+1}^{\lfloor (t+\delta)x_{n_k}\rfloor} \bigg[ I(h_u(\tau),P^m)K\nonumber\\
&&\!\!\!\!\!\!\!\!\times 1_{\{\omega(\tau)=\omega_m,\hat{\textbf{h}}(\tau)
=\hat{\textbf{h}}_i\}}\bigg]\frac{M_{u}}{M_{u}+I_{\max}K}.
\end{eqnarray}

Moreover, we have
\begin{eqnarray}\label{eq:proof34}
&&\!\!\!\!\!\!\!\!\!\!\!\!~~~\lim_{k\rightarrow \infty}\!
\frac{1}{\delta x_{n_k}}\!\!\sum_{\tau=\lfloor tx_{n_k}\rfloor+1}^{\lfloor (t+\delta)x_{n_k}\rfloor}\!\! \bigg[ I(h_u(\tau),P^m)K 1_{\{\omega(\tau)=\omega_m,\hat{\textbf{h}}(\tau)
=\hat{\textbf{h}}_i\}}\bigg]\!\nonumber\nonumber\\
&&\!\!\!\!\!\!\!\!\!\!\!\!=\lim_{k\rightarrow \infty}\!\! \frac{1}{\delta x_{n_k}}\!\!\sum_{\tau=0}^{\lfloor (t+\delta)x_{n_k}\rfloor} \!\! \bigg[ I(h_u(\tau),P^m)K 1_{\{\omega(\tau)=\omega_m,\hat{\textbf{h}}(\tau)
=\hat{\textbf{h}}_i\}}\bigg]\nonumber\\
&&\!\!\!\!\!\!\!\!~~~-\lim_{k\rightarrow \infty}\!\! \frac{1}{\delta x_{n_k}}\!\!
\sum_{\tau=0}^{\lfloor tx_{n_k}\rfloor}\!\! \bigg[ I(h_u(\tau),P^m)K 1_{\{\omega(\tau)=\omega_m,\hat{\textbf{h}}(\tau)
=\hat{\textbf{h}}_i\}}\bigg]\nonumber\\
&&\!\!\!\!\!\!\!\!\!\!\!\!\overset{(a)}{=}\mathbb{E} \{I(h_u,P^m)K|\hat{\textbf{h}}_i\}\nonumber\\
&&\!\!\!\!\!\!\!\!\times\lim_{k\rightarrow \infty} \frac{1}{ \delta x_{n_k}}\bigg[\sum_{\tau=0}^{\lfloor (t+\delta)x_{n_k}\rfloor}1_{\{\omega(\tau)=\omega_m,\hat{\textbf{h}}(\tau)
=\hat{\textbf{h}}_i\}}\nonumber\\
&&-\sum_{\tau=0}^{\lfloor tx_{n_k}\rfloor} 1_{\{\omega(\tau)=\omega_m,\hat{\textbf{h}}(\tau)
=\hat{\textbf{h}}_i\}}\bigg]\nonumber\\
&&\!\!\!\!\!\!\!\!\!\!\!\!\overset{(b)}{=}\mathbb{E}\{I(h_u,P^m)K|\hat{\textbf{h}}_i\} \frac{b_{mi}(t+\delta)-b_{mi}(t)}{ \delta },
\end{eqnarray}
where step $(a)$ follows from the SLLN and the fact that $I(h_u,P^m)$ is \emph{i.i.d.} in the time slots where $\omega(t)=\omega_m$ and $\hat{\textbf{h}}(t)=\hat{\textbf{h}}_i$, and step $(b)$ is due to \eqref{eq:eq20} and \eqref{eq:uoc9}.
Substituting \eqref{eq:fluid_eq3}, \eqref{eq:proof33}, and \eqref{eq:proof34} into \eqref{eq:proof2-lim4}, \eqref{eq:lem3-1} follows.
\subsubsection{Proof of Eq. \eqref{eq:lem4-1}}
It follows from \eqref{eq:20} that
$\sum_{\tau=0}^{t}1_{\{a(\tau) =g \}}=n_g(t)$. Moreover, by the definition of $L_g(n)$, one can obtain that
$$\sum_{n=0}^{n_g(t)}L_g(n)\leq\sum_{\tau=0}^{t}1_{\{s(\tau) =g \}}\leq\sum_{n=0}^{n_g(t)+1}L_g(n).$$
In view of \eqref{eq:blocklength-SLLN}, we obtain
$$\lim_{t\rightarrow \infty}\frac{\sum_{\tau=0}^{t}1_{\{s(\tau) =g \}}}{\sum_{\tau=0}^{t}1_{\{a(\tau) =g \}}}\geq\lim_{t\rightarrow \infty}\frac{\sum_{n=0}^{n_g(t)}L_g(n)}{n_g(t)}=\overline{L}_g,$$
and
$$\lim_{t\rightarrow \infty}\frac{\sum_{\tau=0}^{t}1_{\{s(\tau) =g \}}}{\sum_{\tau=0}^{t}1_{\{a(\tau) =g \}}}\leq\lim_{t\rightarrow \infty}\frac{\sum_{n=0}^{n_g(t)+1}L_g(n)}{n_g(t)}=\overline{L}_g.$$
Hence,
\begin{eqnarray}\label{eq:app1}
\lim_{t\rightarrow \infty}\frac{\sum_{\tau=0}^{t}1_{\{s(\tau) =g \}}}{\sum_{\tau=0}^{t}1_{\{a(\tau) =g \}}}=\overline{L}_g.
\end{eqnarray}

From \eqref{eq:uoc11}, we attain
\begin{eqnarray}
&&\!\!\!\!\!\!\!\!\!~~~\int_0^t\psi_g(\tau)d\tau\nonumber\\
&&\!\!\!\!\!\!\!\!\!=\lim_{k\rightarrow \infty}
\frac{1}{ x_{n_k}}\sum_{\tau=0}^{\lfloor tx_{n_k}\rfloor}M_{g}1_{\{a(\tau) =g \}}\nonumber\\
&&\!\!\!\!\!\!\!\!\!=M_{g}\lim_{k\rightarrow \infty}
\frac{\sum_{\tau=0}^{\lfloor tx_{n_k}\rfloor}1_{\{a(\tau) =g \}}}{\sum_{\tau=0}^{\lfloor tx_{n_k}\rfloor}1_{\{s(\tau) =g \}}}\frac{1}{ x_{n_k}}\sum_{\tau=0}^{\lfloor tx_{n_k}\rfloor}1_{\{s(\tau) =g \}}\nonumber\\
&&\!\!\!\!\!\!\!\!\!\!\!\!\overset{\eqref{eq:app1}}{=}\!M_{g}\frac{1}{\overline{L}_g}\lim_{k\rightarrow \infty}
\frac{1}{ x_{n_k}}\sum_{\tau=0}^{\lfloor tx_{n_k}\rfloor}1_{\{s(\tau) =g \}}\nonumber\\
&&\!\!\!\!\!\!\!\!\!=\frac{M_{g}}{\overline{L}_g}\sum_{m:s^m=g} \sum_{i=1}^E \lim_{k\rightarrow \infty}
\frac{1}{ x_{n_k}}\sum_{\tau=0}^{\lfloor tx_{n_k}\rfloor}1_{\{\omega(\tau)=\omega_m,\hat{\textbf{h}}(\tau)
=\hat{\textbf{h}}_i\}}\nonumber\\
&&\!\!\!\!\!\!\!\!\!\!\!\!\overset{\eqref{eq:uoc9}}{=}\!\frac{M_{g}}{\overline{L}_g}\sum_{m:s^m=g} \sum_{i=1}^E
b_{mi}(t).\nonumber
\end{eqnarray}
Taking the gradient on both sides of this equation, \eqref{eq:lem4-1} is proved.
\subsubsection{Proof of Eq. \eqref{eq:lem5-1}}
From \eqref{eq:uoc13}, we attain
\begin{eqnarray}
&&\!\!\!\!\!\!\!\!\!~~~\int_0^tp(\tau)d\tau\nonumber\\
&&\!\!\!\!\!\!\!\!\!=\lim_{k\rightarrow \infty}
\frac{1}{ x_{n_k}}\sum_{\tau=0}^{\lfloor tx_{n_k}\rfloor}P(t)\nonumber\\
&&\!\!\!\!\!\!\!\!\!=\sum_{m=1}^F \sum_{i=1}^EP^m
\lim_{k\rightarrow \infty}\frac{1}{ x_{n_k}}\sum_{\tau=0}^{\lfloor tx_{n_k}\rfloor}1_{\{\omega(\tau)=\omega_m,\hat{\textbf{h}}(\tau)
=\hat{\textbf{h}}_i\}}\nonumber\\
&&\!\!\!\!\!\!\!\!\!\!\!\!\overset{\eqref{eq:uoc9}}{=}\!\sum_{m=1}^F \sum_{i=1}^E P^m
b_{mi}(t),\nonumber
\end{eqnarray}
Taking the gradient on both sides of this equation, \eqref{eq:lem5-1} is proved.

%
\section{Proof of Lemma \ref{lem5}}\label{App:lem5}

Since $\bm{\lambda}$ is strictly inside $\Lambda$, there exist some parameters $\varepsilon>0$ and $\beta_{mi}\geq0$ such that the following inequalities hold:
\begin{eqnarray}\label{eq:app2}
&&\!\!\!\!\!\!\!\!\!\!\!\!\!\!\!\!\!\!\lambda_u+\varepsilon\leq \!\!\!\sum_{m:s^m=u}\!\sum_{i=1}^E\mathbb{E}\{I(h_u,P^m)|\hat{\textbf{h}}_i\}\frac{M_uK\beta_{mi}\pi_i}{M_u+I_{\max}K},\\
&&\!\!\!\!\!\!\!\!\!\!\!\!\!\!\!\!\!\!\lambda_g+\varepsilon\leq\sum_{m:s^m=g}\sum_{i=1}^E\frac{M_g}{\overline{L}_g}\beta_{mi}\pi_i,\\
&&\!\!\!\!\!\!\!\!\!\!\!\!\!\!\!\!\!\!\sum_{m=1}^F\sum_{i=1}^EP^m\beta_{mi}\pi_i+\varepsilon\leq P_{av}, \label{eq:app4}\\
&&\!\!\!\!\!\!\!\!\!\!\!\!\!\!\!\!\!\!\sum_{m=1}^{M}\beta_{mi}=1. \label{eq:app3}
\end{eqnarray}

Since $q_s$ and $z$ are differentiable for any regular time $t \geq 0$, we can obtain the derivative of $L (t)$ as
\begin{eqnarray}\label{eq:app7}
&&\!\!\!\!\!\!\!\!\!\!\!\!\!~~~\frac{D^+}{dt^+} L(t) \nonumber\\
&&\!\!\!\!\!\!\!\!\!\!\!\!\!\overset{(a)}{=} \sum_{s=1}^{U+G}{q}_s(t)\left[\lambda_s-\psi_s(t)\right]+z(t)\left[p(t)-P_{av}\right]\nonumber\\
&&\!\!\!\!\!\!\!\!\!\!\!\!\!= \sum_{u=1}^{U}{q}_u(t)\!\!\left[\lambda_u\!-\!\!\!\!\sum_{m:s^m=u}\!\sum_{i=1}^E\mathbb{E}\{I(h_u,P^m)|\hat{\textbf{h}}_i\}\frac{M_uK\beta_{mi}\pi_i}{M_u+I_{\max}K}\right]\nonumber\\
&&\!\!\!\!\!\!\!\!\!\!\!\!\!~~~+\sum_{g=U+1}^{U+G}{q}_g(t)\left[\lambda_g-\sum_{m:s^m=g}\!\sum_{i=1}^E\frac{M_g}{\overline{L}_g}\beta_{mi}\pi_i\right]\nonumber\\
&&\!\!\!\!\!\!\!\!\!\!\!\!\!~~~+z(t)\left[\sum_{m=1}^F\!\sum_{i=1}^EP^m\beta_{mi}\pi_i-P_{av}\right]\nonumber\\
&&\!\!\!\!\!\!\!\!\!\!\!\!\! + \sum_{u=1}^{U}{q}_u(t)\bigg[\sum_{m:s^m=u}\!\sum_{i=1}^E\mathbb{E}\{I(h_u,P^m)|\hat{\textbf{h}}_i\}\nonumber\\
&&\times\frac{M_uK\beta_{mi}\pi_i}{M_u+I_{\max}K}-\psi_u(t)\bigg]\nonumber\\
&&\!\!\!\!\!\!\!\!\!\!\!\!\!~~~+\sum_{g=U+1}^{U+G}{q}_g(t)\left[\sum_{m:s^m=g}\!\sum_{i=1}^E\frac{M_g}{\overline{L}_g}\beta_{mi}\pi_i-\psi_g(t)\right]\nonumber\\
&&\!\!\!\!\!\!\!\!\!\!\!\!\!~~~+z(t)\left[p(t)-\sum_{m=1}^F\!\sum_{i=1}^EP^m\beta_{mi}\pi_i\right],
\end{eqnarray}
where $\frac{D^+}{dt^+} L(t)= \lim_{\delta\downarrow0}\frac{L(t+\delta)-L(t)}{\delta}$ and step $(a)$ is due to \eqref{eq:fluid_eq5} and \eqref{eq:fluid_eq12}.

Let us choose $\delta_3>0$ such that $L(t)\geq\delta_1$ implies $\max\{\max_{s\in\{1,\cdots,U+G\}}q_s,z\}\geq \delta_3$. Then, we can conclude from \eqref{eq:app2}-\eqref{eq:app4} that
\begin{eqnarray}\label{eq:app5}
&&\!\!\!\!\!\!\!\!\!\!\!\!\!\sum_{u=1}^{U}{q}_u(t)\!\!\left[\lambda_u\!-\!\!\!\!\sum_{m:s^m=u}\!\sum_{i=1}^E\mathbb{E}\{I(h_u,P^m)|\hat{\textbf{h}}_i\}\frac{M_uK\beta_{mi}\pi_i}{M_u+I_{\max}K}\right]\nonumber\\
&&\!\!\!\!\!\!\!\!\!\!\!\!\!+\sum_{g=U+1}^{U+G}{q}_g(t)\left[\lambda_g-\sum_{m:s^m=g}\!\sum_{i=1}^E\frac{M_g}{\overline{L}_g}\beta_{mi}\pi_i\right]\nonumber\\
&&\!\!\!\!\!\!\!\!\!\!\!\!\!+z(t)\left[\sum_{m=1}^F\!\sum_{i=1}^EP^m\beta_{mi}\pi_i-P_{av}\right]\nonumber\\
&&\!\!\!\!\!\!\!\!\!\leq -\delta_3 \varepsilon \triangleq -\delta_2<0.
\end{eqnarray}

In addition, since RNC-RC chooses scheduling and power allocation decisions according to \eqref{eq96}-\eqref{eq:power}, the following relationship holds:
\begin{eqnarray}\label{eq:app8}
&&\!\!\!\!\!\!\!\!\!c_{mi}(t)\in\arg\max_{\alpha_{mi}:\sum_{m=1}^F\alpha_{mi}=1,\alpha_{mi}\geq0}\nonumber\\
&&\!\!\!\!\!\!\!\!\!\bigg[ \sum_{u=1}^{U}{q}_u(t)\sum_{m:s^m=u} \mathbb{E}\{I(h_u,P^m)K|\hat{\textbf{h}}_i\}\frac{M_u}{M_u+I_{\max}K}\alpha_{mi}~~~\nonumber\\
&&\!\!\!\!\!\!\!\!\!~~+\sum_{g=U+1}^{U+G}\!\!{q}_g(t)\!\!\sum_{m:s^m=g}\frac{M_g}{\overline{L}_g}\alpha_{mi} - z(t)\sum_{m=1}^F P^m \alpha_{mi}\bigg].
\end{eqnarray}
Then, we obtain
\begin{eqnarray}\label{eq:app6}
&&\!\!\!\!\!\!\!\!\!~~~\sum_{u=1}^{U}{q}_u(t) \psi_u(t)+ \sum_{g=U+1}^{U+G}{q}_g(t)\psi_g(t)-z(t)p(t)\nonumber\\
&&\!\!\!\!\!\!\!\!\!\geq  \sum_{i=1}^E \pi_i \bigg[\sum_{u=1}^{U}{q}_u(t)\sum_{m:s^m=u} \mathbb{E}\{I(h_u,P^m)K|\hat{\textbf{h}}_i\}\nonumber\\
&&\!\!\!\!\!\!\!\!\!~~~\times\frac{M_u}{M_u+I_{\max}K}c_{mi}(t)+\!\!\sum_{g=U+1}^{U+G}\!\!{q}_g(t)\!\!\sum_{m:s^m=g}\!\!\frac{M_g}{\overline{L}_g}c_{mi}(t) \nonumber\\
&&\!\!\!\!\!\!\!\!\!~~~ - z(t)\sum_{m=1}^F P^m c_{mi}(t)\bigg]\nonumber\\
&&\!\!\!\!\!\!\!\!\!\geq \sum_{i=1}^E \pi_i \bigg[\sum_{u=1}^{U}{q}_u(t)\sum_{m:s^m=u} \mathbb{E}\{I(h_u,P^m)K|\hat{\textbf{h}}_i\}\nonumber\\
&&\!\!\!\!\!\!\!\!\!~~~\times\frac{M_u}{M_u+I_{\max}K}\beta_{mi}+\sum_{g=U+1}^{U+G}{q}_g(t)\sum_{m:s^m=g}\frac{M_g}{\overline{L}_g}\beta_{mi} \nonumber\\
&&\!\!\!\!\!\!\!\!\!~~~ - z(t)\sum_{m=1}^F P^m \beta_{mi}\bigg].
\end{eqnarray}
where the first inequality is due to Lemma \ref{lem3} and the second inequality is due to \eqref{eq:app8}.
Finally, substituting \eqref{eq:app5} and \eqref{eq:app6} into \eqref{eq:app7}, we attain $\frac{D^+}{dt^+} L(t)\leq -\delta_2$ and the asserted statement is proved.

%% file: appendex1.tex
\section{Proof of Theorem \ref{thm2}}\label{sec:analysis2}
Let $R_v(t)$ be the accumulated mutual information of the unicast file repair flow $v$. The evolutions of $R_v(t)$ are given by
\begin{equation}
R_v(t+1)\! =\!\left\{\!\!\!\begin{array}{l l} R_v(t),&\!\!\!\!\!\!\!\!\!\!\!\!\!\!\!\!\!\!\!\!\!\!\!\!\!\!\!\!\!\!\!\!\!\!\!\!\!\!\!\!\!\!\!\!\!\!\!\!\! \textrm{if}~s(t)\neq v;\\
R_v(t)\!+\!I(h_u(t),P(t))K,\!\!\!\!\!~~&\\
&\!\!\!\!\!\!\!\!\!\!\!\!\!\!\!\!\!\!\!\!\!\!\!\!\!\!\!\!\!\!\!\!\!\!\!\!\!\!\!\!\!\!\!\!\!\!\!\!\!\textrm{if~} s(t)= v \textrm{~and}\\
&\!\!\!\!\!\!\!\!\!\!\!\!\!\!\!\!\!\!\!\!\!\!\!\!\!\!\!\!\!\!\!\!\!\!\!\!\!\!\!\!\!\!\!\!\!\!\!\!\!~R_v(t)\!+\! I(h_u(t),P(t))K\!<\!M_v(n_v(t));\\
0, &\!\!\!\!\!\!\!\!\!\!\!\!\!\!\!\!\!\!\!\!\!\!\!\!\!\!\!\!\!\!\!\!\!\!\!\!\!\!\!\!\!\!\!\!\!\!\!\!\! \textrm{otherwise}.
\end{array}\right.
\end{equation}

Let $\mathcal {X}(t)\triangleq(R_u(t),R_{gj}(t),R_v(t),Q_u(t),Q_g(t),Q_v(t),Z(t),$ $\frac{n_s(t)}{t+1},\frac{T_s(t)}{t+1},\frac{\Upsilon_g(t)}{t+1})$ denote the system state. One can prove that $\mathcal {X}=(\mathcal {X}(t) ,t\geq0)$ is a discrete-time countable Markov chain. Similar to Lemma \ref{lem1} and Lemma \ref{lem2}, we can show that the fluid limit model of $\mathcal {X}$ exists, which can be expressed as $(q_s,a_s,d_s,z,w,v,b_{mi})$.

\begin{lemma}\label{lem7}
The fluid limit functions satisfy the properties:
{\small\begin{eqnarray}
&&\!\!\!\!\!\!\!\!\!\!\!\!\psi_u(t)\geq \sum_{m:s^m=u} \sum_{i=1}^E \mathbb{E}\{I(h_u,P^m)K|\hat{\mathbf{h}}_i\}\frac{M_uc_{mi}(t)\pi_i}{M_s+I_{\max}K},~\label{eq:lem7-2}\\
&&\!\!\!\!\!\!\!\!\!\!\!\!\psi_g(t)= \sum_{m:s^m=g} \sum_{i=1}^E \frac{M_g}{\overline{L}_g(l(g))}c_{mi}(t)\pi_i,\label{eq:lem7-3}\\
&&\!\!\!\!\!\!\!\!\!\!\!\!\psi_v(t)\geq (1-\eta_{gj})\!\!\!\!\!\!\!\sum_{m:s^m=v(g,j)}\! \sum_{i=1}^E \mathbb{E}\{I(h_{gj},P^m)K|\hat{\mathbf{h}}_i\}\frac{M_gc_{mi}(t)\pi_i}{M_g+I_{\max}K}\nonumber\\
&&+\eta_{gj}\sum_{m:s^m=g} \sum_{i=1}^E \frac{M_g}{\overline{L}_g(l(g))}c_{mi}(t)\pi_i,\label{eq:lem7-5}\\
&&\!\!\!\!\!\!\!\!\!\!\!\!p(t)= \sum_{m=1}^F \sum_{i=1}^E P^mc_{mi}(t)\pi_i,\label{eq:lem7-6}
\end{eqnarray}}
$\!\!$for all $u\in\{1,\cdots,U\}$, $g\in\{U+1,\cdots,U+G\}$, and $v=v(g,j)$.
\end{lemma}

\begin{proof}
Equations \eqref{eq:lem7-2}, \eqref{eq:lem7-3}, and \eqref{eq:lem7-6} follows from Lemma \ref{lem3}. Therefore, we only need to prove \eqref{eq:lem7-5}.

\begin{eqnarray}\label{eq:233}
&&\!\!\!\!\!\!\!\!\!~~~\int_0^t\psi_v(\tau)d\tau\nonumber\\
&&\!\!\!\!\!\!\!\!\!=\lim_{k\rightarrow \infty}
\frac{1}{ x_{n_k}}\sum_{\tau=0}^{\lfloor tx_{n_k}\rfloor}\left[M_{v}(n_v(\tau))1_{\{a(\tau) =v \}}\right.\nonumber\\
&&\left.+R^\ast_{gj}(n_g(\tau)) 1_{\{a(\tau) =g \}}\right]\nonumber\\
&&\!\!\!\!\!\!\!\!\!=\lim_{k\rightarrow \infty}
\frac{1}{ x_{n_k}}\sum_{\tau=0}^{\lfloor tx_{n_k}\rfloor}M_{v}(n_v(\tau))1_{\{a(\tau) =v \}}\nonumber\\
&&+\lim_{k\rightarrow \infty}\frac{1}{ x_{n_k}}\sum_{\tau=0}^{\lfloor tx_{n_k}\rfloor}R^\ast_{gj}(n_g(\tau)) 1_{\{a(\tau) =g \}}
\end{eqnarray}
By following the proof idea in Appendix \ref{sec:proof_lem3_1} and using the fact that
\begin{eqnarray}
\lim_{N\rightarrow \infty } \frac{\sum\limits_{n=1}^NM_{v}(n)}{N}=(1-\eta_{gj})M_g.
\end{eqnarray}
we can obtain
\begin{eqnarray}\label{eq:234}
&&\!\!\!\!\!\!\!\!\!~~~\lim_{k\rightarrow \infty}
\frac{1}{ x_{n_k}}\sum_{\tau=0}^{\lfloor tx_{n_k}\rfloor}M_{v}(n_v(\tau))1_{\{a(\tau) =v \}}\nonumber\\
&&\!\!\!\!\!\!\!\!\!=\sum_{m:s^m=v(g,j)}\!\sum_{i=1}^E\mathbb{E}\{I(h_{gj},P^m)K|\hat{\mathbf{h}}_i\}\nonumber\\
&&\times\frac{(1-\eta_{gj})M_g}{(1-\eta_{gj})M_g+I_{\max}K}b_{mi}(t).
\end{eqnarray}
On the other hand, using \eqref{eq:8}, we attain
\begin{eqnarray}\label{eq:235}
&&\!\!\!\!\!\!\!\!\!~~~\lim_{k\rightarrow \infty}\frac{1}{ x_{n_k}}\sum_{\tau=0}^{\lfloor tx_{n_k}\rfloor}R^\ast_{gj}(n_g(\tau)) 1_{\{a(\tau) =g \}}\nonumber\\
&&\!\!\!\!\!\!\!\!\!=\eta_{gj}M_g\lim_{k\rightarrow \infty}\frac{1}{ x_{n_k}}\sum_{\tau=0}^{\lfloor tx_{n_k}\rfloor} 1_{\{a(\tau) =g \}}\nonumber\\
&&\!\!\!\!\!\!\!\!\!=\eta_{gj}\frac{M_g}{\overline{L}_g(l(g))}\lim_{k\rightarrow \infty}
\frac{1}{ x_{n_k}}\sum_{\tau=0}^{\lfloor tx_{n_k}\rfloor}1_{\{s(\tau) =g \}}\nonumber\\
&&\!\!\!\!\!\!\!\!\!=\eta_{gj}\frac{M_{g}}{\overline{L}_g(l(g))}\sum_{m:s^m=g} \sum_{i=1}^E \lim_{k\rightarrow \infty}
\frac{1}{ x_{n_k}}\!\sum_{\tau=0}^{\lfloor tx_{n_k}\rfloor}\!\!1_{\{\omega(\tau)=\omega_m,\hat{\textbf{h}}(\tau)
=\hat{\textbf{h}}_i\}}\nonumber\\
&&\!\!\!\!\!\!\!\!\!\!\!\!\overset{\eqref{eq:uoc9}}{=}\!\eta_{gj}\frac{M_{g}}{\overline{L}_g(l(g))}\sum_{m:s^m=g} \sum_{i=1}^E
b_{mi}(t).
\end{eqnarray}
Substituting \eqref{eq:234} and \eqref{eq:235} into \eqref{eq:233} and taking the gradient over $t$, \eqref{eq:lem7-5} follows.
\end{proof}
Theorem \ref{thm2} can be proved by employing the same arguments used in Appendix \ref{sec:app_sta}, except that Lemma \ref{lem3} is replaced by Lemma \ref{lem7}.